\documentclass[10pt,journal,compsoc]{IEEEtran}
\usepackage{units}
\usepackage{algorithm}
\usepackage{algorithmicx}
\usepackage{algpseudocode}
\algnewcommand{\LineComment}[1]{\State \(\triangleright\) #1}

\usepackage{amsthm}
\usepackage{amsmath}
\usepackage{graphicx}
\usepackage{url}

\theoremstyle{definition}
\newtheorem{defn}{\protect\definitionname}
\theoremstyle{plain}
\newtheorem{prop}{\protect\propositionname}
  
\providecommand{\definitionname}{Definition}
\providecommand{\propositionname}{Proposition}

\begin{document}
\title{t-Closeness through Microaggregation: Strict Privacy
with Enhanced Utility Preservation}

%

\author{Jordi Soria-Comas,
        Josep Domingo-Ferrer,~\IEEEmembership{Fellow, IEEE,}
	David S\'anchez and Sergio Mart\'{\i}nez
\thanks{The authors are with the UNESCO Chair in Data Privacy,
Department of Computer Engineering and Mathematics,
Universitat Rovira i Virgili, Av. Pa\"{\i}sos Catalans 26,
E-43007 Tarragona, Catalonia. E-mail \{jordi.soria, josep.domingo, david.sanchez, sergio.martinezl\}@urv.cat}}%

%
%

\markboth{IEEE TRANSACTIONS ON KNOWLEDGE AND DATA ENGINEERING,~Vol.~?, No.~?, Month~YYYY}%
{Jordi Soria-Comas \MakeLowercase{\textit{et al.}}: $t$-Closeness through Microaggregation}
%



\IEEEtitleabstractindextext{%
\begin{abstract}
Microaggregation is a technique for disclosure limitation 
aimed at protecting
the privacy of data subjects in microdata releases. 
It has been used as an alternative
to generalization and suppression to generate $k$-anonymous
data sets, where the identity of each subject is hidden within
a group of $k$ subjects. 
Unlike generalization, microaggregation perturbs the data and this
additional masking freedom allows improving data
utility in several ways, 
such as increasing data granularity, reducing the impact
of outliers and avoiding discretization of numerical data.
$k$-Anonymity, on the other side, does
not protect against attribute disclosure, which occurs if the variability
of the confidential values in a group of $k$ subjects is too small.
To address this issue, several refinements of $k$-anonymity have been
proposed, among which $t$-closeness stands out as providing
one of the strictest privacy guarantees. 
Existing algorithms
to generate $t$-close data sets are based on generalization
and suppression (they are extensions of $k$-anonymization algorithms
based on the same principles).
This paper proposes and shows how to use microaggregation to generate
$k$-anonymous $t$-close data sets. The advantages
of microaggregation are analyzed, and then several microaggregation 
algorithms 
for $k$-anonymous $t$-closeness are presented and empirically evaluated.
\end{abstract}
}

\maketitle

%
\IEEEpeerreviewmaketitle

\section{Introduction}
\label{sec1}
%
%
%
%
\IEEEPARstart{G}{enerating} an anonymized data set that is suitable for public
release is essentially a matter of finding a good equilibrium between
disclosure risk and information loss. Releasing the original data
set provides the highest utility to data users but incurs the greatest
disclosure risk for the subjects in the data set. On the contrary,
releasing random data incurs no risk of disclosure but provides
no utility. 

$k$-Anonymity~\cite{Samarati1998Protecting,Sweeney2002kAnonymity}
is the oldest among 
the so-called syntactic privacy models.  
Models in this class address the trade-off between privacy and
utility by requiring the anonymized data set to follow a specific
pattern that is known to limit the risk of disclosure. Yet, the method
to be used to generate such an anonymized data set is not specified by the
privacy model and must be selected to maximize data utility (because
satisfying the model already ensures privacy). $k$-Anonymity,
in particular, seeks to make record re-identification unfeasible by
hiding each subject within a group of $k$ subjects. To this end,
$k$-anonymity requires each record in the anonymized data set to
be indistinguishable from another $k-1$ records as far as the
quasi-identifier attributes are
 concerned (see Section~\ref{sec:background} for a classification
of attributes into identifiers, quasi-identifiers, confidential 
attributes and other attributes).

Although $k$-anonymity protects against identity disclosure (the
subject to whom a record corresponds 
cannot be successfully re-identified with probability
greater than $1/k$), disclosure can still happen if the variability
of the confidential attribute values 
in the group of $k$ records is small. This
is known as \emph{attribute disclosure}. Several refinements of the
$k$-anonymity model have been proposed to protect against attribute
disclosure;
they all seek to
guarantee at least a certain amount of variability of
the confidential attribute values
within each group of indistinguishable records. In this paper we focus
on the notion of $t$-closeness~\cite{Li2007t-Closeness}, whose
privacy guarantee is probably the strictest
among $k$-anonymity-like models. In fact,
$t$-closeness has been
shown in~\cite{Soria2013differential,DomingoSoria15} to be related
to the major alternative
to $k$-anonymity-like models, namely $\varepsilon$-differential 
privacy~\cite{Dwork06}.
$t$-Closeness requires that the distribution
of the confidential attribute values 
within each group of indistinguishable records be
similar to the distribution of the confidential 
attribute values in the entire data set.

The dominant approach to obtain an anonymized
data set satisfying $k$-anonymity or any of its refinements  
is based on generalization (recoding) and suppression. The goal of 
generalization-based approaches is to find the minimal generalization
that satisfies the requirements of the underlying privacy model. 
These algorithms can be adapted to the above-mentioned 
$k$-anonymity refinements: it is simply a matter of introducing
the additional constraints of the target privacy model when checking
whether a specific generalization is viable.

Generalization-based approaches suffer from 
some drawbacks identified in~\cite{Domingo2005Ordinal} and reviewed
in Section~\ref{sec:comparison} below.
Microaggregation was shown in~\cite{Domingo2005Ordinal} 
to be an alternative approach to 
generate $k$-anonymous
data sets while avoiding some of these drawbacks. 
%

\subsection*{Contribution and plan of this paper}

A first contribution of this paper is to identify the 
strong points of   
microaggregation to achieve $k$-anonymous $t$-closeness.
The second contribution consists of 
three new microaggregation-based algorithms for $t$-closeness, which
are presented and evaluated.
 
In Section~\ref{sec:background}
we review some concepts used throughout the paper: 
$k$-anonymity, $t$-closeness, recoding/generalization and
microaggregation. In Section~\ref{sec:comparison}, we
identify the advantages of microaggregation
over generalization/suppression for $k$-anonymity and hence
for $t$-closeness as well; then we 
sketch three microaggregation-based algorithms for $t$-closeness
that are detailed in the next sections.
Section~\ref{sub:micro_merge} presents an algorithm for 
$t$-closeness based on standard microaggregation followed
by cluster merging. Section~\ref{sub:k-anonymity-first}
presents an algorithm that embeds $t$-closeness into the microaggregation
process: each cluster is generated to satisfy $k$-anonymity and
then it is refined to achieve $t$-closeness.
Section~\ref{sub:t-closeness-first} also embeds $t$-closeness
in the microaggregation process, but in this case each cluster
is generated to satisfy simultaneous $k$-anonymity and $t$-closeness
from the very beginning. 
In Section~\ref{sec:empirical} we evaluate the 
previously proposed algorithms on real data sets. Conclusions
are gathered in Section~\ref{sec:conclusions}.

\section{Background}
\label{sec:background}

A microdata set can be modeled as a table where each row contains
data on a different subject and each column contains information
about a specific attribute. Let $T(A_{1},\ldots,A_{m})$ be a microdata
set with $n$ records $r_{1},\ldots,r_{n}$, each of them with information
about attributes $A_{1},\ldots,A_{m}$.

The attributes in a microdata set can be classified according to their
disclosiveness into several (perhaps non-disjoint) classes
(see~\cite{Hundepool2012sdc}~for more details on the following
classification):
identifiers, quasi-identifiers, confidential attributes, and non-confidential attributes.


Disclosure risk limitation (a.k.a. statistical disclosure control) 
seeks to restrict the capability 
of an intruder
with access to the released data set to
associate a piece of confidential
information to a specific subject in the data set.  
To this end, a masked version $T'(A_{1},\ldots,A_{n})$
of the original data set $T(A_{1},\ldots,A_{n})$ is released. We use
the term \emph{anonymized data set} to refer to $T'(A_{1},\ldots,A_{n})$. 

\subsection{$k$-Anonymity}

An intruder re-identifies a record in an anonymized data set when
he can determine the identity of the subject
to whom the record corresponds.
In case of re-identification, the intruder can 
associate the values of the confidential attributes
in the re-identified record to the 
identity of the subject, thereby violating the subject's privacy. 

$k$-Anonymity~\cite{Samarati1998Protecting,Sweeney2002kAnonymity}
seeks to limit the capability of the intruder to perform successful
re-identifications. 

\begin{defn}[$k$-anonymity]
	Let $T$ be a data set and $QI_{T}$ be the 
set of quasi-identifier attributes in it. 
$T$ is said to satisfy $k$-anonymity if, for
each combination of values of the quasi-identifiers in $QI_T$, at least
$k$ records in $T$ share that combination.
\end{defn}
In a $k$-anonymous data set, no subject's identity 
can be linked (based on the quasi-identifiers) 
to less than $k$ records. Hence, the probability of correct
re-identification is, at most, $1/k$.
In what follows, 
we use the terms \emph{$k$-anonymous
       group} or \emph{equivalence class} to refer to a set of records that
share the quasi-identifier values.

%

\subsection{$t$-Closeness}

Even though $k$-anonymity protects against identity disclosure, 
it is a well-known fact that $k$-anonymous data sets are vulnerable
to attribute disclosure. Attribute disclosure occurs when the variability
of a confidential attribute within an equivalence class is too low.
In that case, being able to determine the equivalence class of a subject 
may reveal too much information about the
confidential attribute value of that subject.

Several refinements of $k$-anonymity have been proposed to deal
with attribute disclosure. For example, $p$-sensitive $k$-anonymity~\cite{Truta2006pSensitive},
$l$-diversity~\cite{Machanavajjhala2007lDiversity}, $t$-closeness~\cite{Li2007t-Closeness},
and $(n,t)$-closeness~\cite{Li2010Closeness}. As explained
in Section~\ref{sec1}, in this paper we
focus on $t$-closeness because of its strict privacy guarantee 
(although the methods we propose are easily
adaptable to $(n,t)$-closeness). 

$t$-Closeness seeks to limit the amount of information that
an intruder can obtain about the
confidential attribute of any specific subject. 
To this end, $t$-closeness requires 
the distribution of the confidential
attributes within each of the equivalence classes 
to be similar to their distribution
in the entire data set. 
\begin{defn}
	An equivalence class is said to satisfy $t$-closeness if the distance
	between the distribution of the confidential attribute in this class
	and the distribution of the attribute in the whole data set is no more
	than a threshold $t$. A data set (usually 
a $k$-anonymous data set) is said to satisfy $t$-closeness if all
	equivalence classes in it satisfy $t$-closeness.
\end{defn}

The specific distance used between distributions is central to evaluate 
$t$-closeness, but the original definition does not advocate any specific
distance. The Earth Mover's distance (EMD)~\cite{Rubner2000Earth} 
is the most common choice (and the one we will adopt in this paper),
although other distances have also been 
explored~\cite{Rebollo,Soria2013differential,DomingoSoria15}.
$EMD(P,Q)$ measures the cost of transforming one distribution $P$ 
into another distribution $Q$ by moving probability mass.
EMD is computed as the minimum transportation
cost from the bins of $P$ to the bins of $Q$, so it 
depends on how much mass is moved and how far it is moved.
For numerical attributes the distance between two bins is
based on the number of bins between them. If the numerical
attribute takes values $\{v_1,v_2,\ldots,v_m\}$, where $v_i<v_j$
if $i<j$, then $ordered\_distance(v_i,v_j)=|i-j|/(m-1)$. 
Now, if $P$ and $Q$ are distributions over $\{v_1,v_2,\ldots,v_m\}$
that, respectively, assign probability $p_i$ and $q_i$ to $v_i$,
then the EMD for the ordered distance can be computed as
\[EMD(P,Q) = \frac{1}{m-1} \sum_{i=1}^{m} \left| \sum_{j=1}^{i} p_j-q_j \right|\]

\subsection{Microaggregation}
\label{micro}

Microaggregation is a family of perturbative methods
for statistical disclosure control of microdata releases. 
One-dimensional microaggregation was
introduced in~\cite{Defays1992Aggregates} 
and multi-dimensional microaggregation was proposed
and formalized in~\cite{Domingo02}.
The latter is the one that is useful
for $k$-anonymity and $t$-closeness.
It consists
of the following two steps:
\begin{itemize}
	\item {\em Partition:} The records in the original data set are partitioned
	into several clusters, each of them containing at least $k$ records.
	To minimize the information loss, records in each cluster should be
	as similar as possible.
	\item {\em Aggregation:} An aggregation operator is used to summarize the
	data in each cluster 
	and the original records are replaced by the aggregated output.
 	For numerical data, one can use the mean as aggregation operator;
	for categorical data, one can resort to the median or some 
	other average operator defined in terms of an ontology ({\em e.g.} 
see~\cite{Domingo13}).
\end{itemize}
The partition and aggregation steps produce some information loss.
The goal of microaggregation is to minimize the information loss according
to some metric. A common information loss metric is the SSE (sum of
squared errors). 
When using SSE on 
numerical attributes, the mean is a sensible choice
as the aggregation operator, because for any given partition 
it minimizes SSE in the aggregation step;
the challenge thus is to come up with a partition that minimizes the
overall SSE.
Finding an optimal partition in multi-dimensional
microaggregation is an NP-hard problem~\cite{Oganian01Complexity};
therefore, heuristics are employed to obtain an approximation with
reasonable cost.


The limitations to re-identification imposed by $k$-anonymity can
be satisfied without aggregating the values of the quasi-identifier
attributes within each equivalence class after the partition step. 
It is less utility-damaging to break the relation 
between quasi-identifiers and confidential attributes 
while preserving the original values of the quasi-identifiers.
This is the approach to attain $k$-anonymity-like guarantees 
taken in~\cite{Xiao2006Anatomy,Soria2012Probabilistic}.


\section{Related Work}
\label{sec:related_work}
Same as for $k$-anonymity, the most common way to attain $t$-closeness
is to use generalization and suppression. In fact, the algorithms
for $k$-anonymity based on those principles 
can be adapted to yield $t$-closeness by adding the
$t$-closeness constraint in the search for a feasible
minimal generalization: in~\cite{Li2007t-Closeness} the Incognito
algorithm and in~\cite{Li2010Closeness} the Mondrian algorithm are
respectively adapted to $t$-closeness. 
SABRE~\cite{Cao} is another interesting approach specifically designed 
for $t$-closeness. In SABRE the data set is first partitioned into 
a set of buckets and 
then the equivalence classes are generated 
by taking an appropriate number of records 
from each of the buckets. Both the buckets and the number of records from
each bucket that are included in each equivalence class 
are selected with $t$-closeness
in mind. One of the algorithms proposed in our paper uses a similar
principle. However,
the buckets in SABRE are generated in an iterative greedy manner
which may yield more buckets than our algorithm (which 
analytically determines
the minimal number of required buckets). 
A greater number of buckets
leads to equivalence classes with more records and, thus, to more information loss.

In~\cite{Rebollo} an approach to attain $t$-closeness-like privacy is
 proposed
which, unlike the methods based on generalization/suppression, is perturbative. 
Also, \cite{Rebollo} guarantees the threshold $t$ only on 
average and uses a distance 
other than EMD. Another computational approach to $t$-closeness 
is presented 
in~\cite{DomingoSoria15,Soria2013differential} which aims at 
connecting $t$-closeness and
differential privacy; \cite{DomingoSoria15,Soria2013differential} also
use a distance different from EMD but their method
is non-perturbative (the truthfulness of the data is preserved).

Most of the approaches to attain $t$-closeness have been designed to
preserve the truthfulness of the data. In this paper we evaluate the use of
microaggregation, a perturbative masking technique. 
In $k$-anonymity the relation between the quasi-identifiers and the
confidential data is broken by making records in the anonymized data
set indistinguishable in terms of quasi-identifiers within a group
of $k$ records. Microaggregation, when performed on the 
projection on quasi-identifier
attributes, produces a $k$-anonymous data set~\cite{Domingo2005Ordinal}.
Microaggregation was also used for $k$-anonymity without naming it 
in~\cite{Li2008LocalRecoding}: clustering was used with the 
additional requirement that each cluster must have $k$ or more 
records.

While microaggregation has been proposed to satisfy
another refinement of $k$-anonymity 
($p$-sensitive $k$-anonymity, \cite{Solanas2008microaggregation}),
no attempt has been made to use it for $t$-closeness.

\section{$k$-Anonymity/$t$-closeness and microaggregation}
\label{sec:comparison}



Microaggregation has several 
advantages over generalization/recoding for $k$-anonymity that are 
mostly related to data utility preservation:
\begin{itemize}
\item Global recoding may recode some records that
do not need it, hence causing extra information loss.
On the other hand, local recoding makes data analysis more
complex, as values corresponding to various different
levels of generalization may co-exist in the anonymized data.
Microaggregation is free from either drawback.
\item Data generalization usually results in a significant 
loss of granularity, because input values can only be replaced 
by a reduced set of generalizations, which 
are more constrained as one moves up in the hierarchy. 
Microaggregation, on the other hand,
does not reduce the granularity of values, because they are 
replaced by numerical or categorical averages.
\item If outliers are present in the input data, 
the need to generalize them
results in very coarse generalizations and, thus, in a high loss of information.
For microaggregation, the influence of an outlier in the 
calculation of averages/centroids
is restricted to the outlier's equivalence class and hence is less noticeable.
\item For numerical attributes, generalization discretizes 
input numbers to numerical ranges and thereby changes 
the nature of data from continuous to discrete.
In contrast, microaggregation 
maintains the continuous nature of numbers.
\end{itemize} 

In~\cite{Samarati1998Protecting,Sweeney2002kAnonymity} it was proposed
to combine local suppression with recoding to reduce
the amount of recoding. Local
suppression has several drawbacks:
\begin{itemize}
\item It is not known how to optimally combine generalization and 
local suppression.
\item There is no agreement in the literature on how suppression
should be performed:
one can suppress at the record level (entire record suppressed),
or suppress particular attributes in some records; furthermore, 
suppression can be done by either blanking a value or replacing it 
by a neutral value ({\em i.e.} some kind of average).
\item Last but not least, and no matter how suppression is performed,
it complicates data analysis (users need to resort to software
dealing with censored data).
\end{itemize}  

Some of the above downsides of generalization and suppression motivated
proposing microaggregation for $k$-anonymity 
in~\cite{Domingo2005Ordinal}. They also justify that we   
investigate here the use of microaggregation 
for $t$-closeness.



The adaptation of microaggregation  for $k$-anonymity
was pretty straightforward: by applying the microaggregation algorithm
(with minimum cluster size $k$) to the quasi-identifiers one generates
groups of $k$ records that share the quasi-identifier values (the
aggregation step replaces the original quasi-identifiers by the cluster
centroid). In microaggregation one seeks to maximize the homogeneity
of records within a cluster, which is beneficial for the utility of
the resultant $k$-anonymous data set. 

In $t$-closeness one has   
the additional constraint that the distance between 
the distribution of the confidential attribute within each of the clusters 
(generated by microaggregation)
and the distribution in the entire data set must be less than $t$.
This makes attaining $t$-closeness more complex, because we have
to reconcile the possibly conflicting goals of maximizing the 
within-cluster homogeneity of the quasi-identifiers and 
fulfilling the condition on the distance between the distributions
of the confidential attributes.

In the next three sections, we propose three different algorithms 
to reconcile these conflicting goals.
The first algorithm is based on performing microaggregation in 
the usual way, 
and then merging clusters as much as needed 
to satisfy the $t$-closeness condition. This first algorithm
is simple and it can be combined with any microaggregation algorithm,
yet it may perform poorly regarding utility because clusters 
may end up being quite large.
The other algorithms modify the
microaggregation algorithm for it to take $t$-closeness into account,
in an attempt to improve the utility of the anonymized data set. Two
variants are proposed: $k$-anonymity-first (which generates each cluster
based on the quasi-identifiers and then refines it to satisfy $t$-closeness)
and $t$-closeness-first (which generates each cluster 
based on both quasi-identifier
attributes and confidential attributes, so that it satisfies $t$-closeness
by design from the very beginning). 

\section{Standard microaggregation and merging\label{sub:micro_merge}}

Generating a $t$-close data set via generalization is essentially
an optimization problem: one must find a minimal generalization 
that satisfies $t$-closeness.
A common way to find a solution is to iteratively
generalize one of the attributes (selected according to some quality
criterion) until the resulting data set satisfies $t$-closeness. Our 
first proposal to attain $t$-closeness via microaggregation follows
a similar approach. We microaggregate and then 
merge clusters of records in the microaggregated data set;
we use the distance between the quasi-identifiers of the
microaggregated clusters as the quality criterion
to select which groups are to be merged.  

Initially, the microaggregation algorithm is run on the quasi-identifier
attributes of the original data set; this step produces a $k$-anonymous
data set. Then, clusters of microaggregated records are merged until
$t$-closeness is satisfied. We iteratively improve the level of $t$-closeness
by: i) selecting the cluster whose confidential attribute 
distribution is most different from the confidential
attribute distribution in the entire data set (that is,
the cluster farthest from 
satisfying $t$-closeness); and ii) merging it with the cluster 
closest to it in terms of quasi-identifiers. 
See Algorithm~\ref{alg:micro_merge}
for a detailed description of the algorithm.

\begin{algorithm}
	\protect\caption{\label{alg:micro_merge}$t$-Closeness through microaggregation and
		merging of microaggregated groups of records.}
		
	\begin{algorithmic}[0]
		\State {\bf Data:} $X$: original data set
		\State \hspace{1.11cm}$k$: minimum cluster size
		\State \hspace{1.11cm}$t$: t-closeness level
		\State {\bf Result} Set of clusters satisfying $k$-anonymity and $t$-closeness
		\vspace{0.3cm}
		\State $X'$=microaggregation($X$, $k$)\;
		\While{$EMD(X',X)>t$}
			\State$C$ = cluster in $X'$ with the greatest $EMD$ to $X$
			\State$C'$ = cluster in $X'$ closest to $C$ in terms of QIs
			\State$X'$ = merge $C$ and $C'$ in $X'$
		\EndWhile\label{f}
		\State {\bf return} $X'$\;
	\end{algorithmic}

\end{algorithm}

Note that Algorithm~\ref{alg:micro_merge} always returns a $t$-close
data set. In the worst case, all clusters are eventually 
merged into a single one and the EMD becomes zero.

The computational cost of Algorithm~\ref{alg:micro_merge} is the sum of the cost of the initial
microaggregation and the cost of merging clusters. 
Although optimal multivariate microaggregation is NP-hard, 
several heuristic approximations exist with quadratic cost 
on the number $n$ of records
of $X$ (e.g. MDAV~\cite{Domingo2005Ordinal}, V-MDAV~\cite{Solanas2006Vmdav}). For the merging
part, the fact that computing the EMD for numerical data has linear cost turns the merging 
quadratic. 
More precisely, the cost of Algorithm~\ref{alg:micro_merge} is 
$\max\{\mathcal{O}(microaggregation),n^2/k\}$. If MDAV is used for the microaggregation, the 
cost is $\mathcal{O}(n^2/k)$.

\section{$t$-Closeness aware microaggregation: $k$-anonymity-first\label{sub:k-anonymity-first}}

Algorithm~\ref{alg:micro_merge} consists of two clearly defined steps: first
microaggregate and then merge clusters until $t$-closeness
is satisfied. In the microaggregation step any standard microaggregation
algorithm can be used because the enforcement of $t$-closeness
takes place only after microaggregation is complete. As a result,
the algorithm is quite clear, but the utility of the anonymized data
set may be far from optimal. If, instead of deferring the enforcement of
$t$-closeness to the second step, we make the microaggregation algorithm
aware of the $t$-closeness constraints at the time 
of cluster formation,
the size of the resulting clusters and also information loss 
can be expected to be smaller.

Algorithm~\ref{alg:k-anonymity-first} microaggregates according
to the above idea.
It initially generates a cluster of size $k$ based on the quasi-identifier
attributes. Then the cluster is iteratively refined until $t$-closeness
is satisfied. In the refinement, the algorithm checks whether 
$t$-closeness is satisfied
and, if it is not, it selects the closest record not in the cluster based
on the quasi-identifiers and swaps it with a record in the cluster
selected
so that the EMD to the distribution of the entire data set is minimized. 

\begin{algorithm}
	\protect\caption{\label{alg:k-anonymity-first}$k$-Anonymity-first $t$-closeness aware microaggregation algorithm.}
	
	\begin{algorithmic}[0]
		\Function{$k$-Anonymity-first}{}
			\State {\bf Data:} $X$: original data set
			\State \hspace{1.11cm}$k$: minimum cluster size
			\State \hspace{1.11cm}$t$: t-closeness level
			\State {\bf Result} Set of clusters satisfying $k$-anonymity and $t$-closeness
			\vspace{0.3cm}
			\State $Clusters = \emptyset$
			\State $X' = X$
			\While{$X'\ne \emptyset$}
				\State$x_a$ = average record of $X'$
				
				\State$x_0$ = most distant record from $x_a$ in $X'$
				\State$C$ = GenerateCluster($x_0$, $X'$, $X$, $k$, $t$)
				\State$X'=X'\setminus C$
				\State$Clusters = Clusters \cup \{C\}$
				
				\If{$X'\ne \emptyset$}
					\State$x_1$ = most distant record from $x_0$ in $X'$
					\State$C$ = GenerateCluster($x_1$, $X'$, $X$, $k$, $t$)
					\State$X'=X'\setminus C$
					\State$Clusters = Clusters \cup \{C\}$
				\EndIf
			\EndWhile
		\State {\bf return} $Clusters$
		\EndFunction
		
		\Function{GenerateCluster($x$, $X'$, $X$, $k$, $t$)}{}
			\State {\bf Data:} $x$: source record for the cluster
			\State \hspace{1.11cm}	$X'$: remaining unclustered records of $X$
			\State \hspace{1.11cm}	$X$: original data set
			\State \hspace{1.11cm}	$k$: minimum cluster size
			\State \hspace{1.11cm}	$t$: desired $t$-closeness level
			\State {\bf Result} $t$-close cluster of $k$ (or more) records
			
			\If{$|X'|<2k$}
				\State $C = X'$
			\Else
				\State $C$ = $k$ closest records to $x$ in $X'$ (including $x$ itself)
				\State $X'=X'\setminus C$
				\While{$EMD(C,X)>t$ and $X'\ne\emptyset$}
					\State $y$ = record in $X'$ that is closest to $x$
					\State $y'$ = record $C$ that minimizes $EMD(C\cup\{y\}\setminus\{y'\},X)$
					\If{$EMD(C\cup\{y\}\setminus\{y'\},X) < EMD(C,X)$}
						\State $C$=$C\cup\{y\}\setminus\{y'\}$
					\EndIf
					\State $X'=X'\setminus\{y\}$
				\EndWhile
			\EndIf
			\State {\bf return} $C$
		\EndFunction
		
	\end{algorithmic}

\end{algorithm}

Instead of replacing the records already added to a cluster, we could have opted for
adding additional records until $t$-closeness is satisfied. This latter approach was 
discarded because it led to large clusters when the dependence between 
quasi-identifiers
and confidential attributes is high.
In this case,
clusters homogeneous in terms of quasi-identifiers tend
to be homogeneous in terms of confidential attributes, 
so the within-cluster distribution of the confidential attribute
differs from its distribution in the entire data set unless
the cluster is (nearly) as big as the entire data set. 

It may happen that the records in 
the data set are exhausted before $t$-closeness
is satisfied. This is most likely when the number of remaining
unclustered records is small (for instance, when the last cluster
is formed). 
Thus, {\em Algorithm~\ref{alg:k-anonymity-first} alone
cannot guarantee that $t$-closeness is satisfied.
A way to circumvent this shortcoming is 
to use Algorithm~\ref{alg:k-anonymity-first}
as the microaggregation function in Algorithm~\ref{alg:micro_merge}.}
By taking into account $t$-closeness at the time 
of cluster formation (as Algorithm~\ref{alg:k-anonymity-first} does),
the number of cluster mergers in Algorithm~\ref{alg:micro_merge} 
can be expected to be small and, therefore, the utility 
of the resulting anonymized data set can be expected to be reasonably
good.

Algorithm~\ref{alg:k-anonymity-first} makes an intensive use of the EMD distance.
Due to this and to the cost of computing EMD, 
Algorithm~\ref{alg:k-anonymity-first}
may be rather slow. More precisely, it has order $\mathcal{O}(n^3/k)$ in 
the worst case, and order
$\mathcal{O}(n^2/k)$ in the best case (when no record swaps are required).

\section{$t$-Closeness aware microaggregation: $t$-closeness-first\label{sub:t-closeness-first}}

In Section~\ref{sub:k-anonymity-first} we modified the microaggregation
algorithm for it to build the clusters in a $t$-closeness aware manner.
The clustering algorithm, however, kept the focus on the quasi-identifiers
(records were selected based on the quasi-identifiers) and did not
guarantee that every cluster satisfies $t$-closeness. The algorithm 
proposed in this section prioritizes the confidential
attribute, thereby making it possible to guarantee that all clusters
satisfy $t$-closeness.

We assume in this section that the values of the confidential attribute(s)
can be ranked, that is, be ordered in some
way. For numerical or categorical ordinal attributes, 
ranking is straightforward. Even for categorical nominal
attributes, the ranking assumption is less restrictive
than it appears, because
the same distance metrics that are used to microaggregate
this type of attributes  
can be used to rank them ({\em e.g.} 
the marginality distance in~\cite{Domingo13,SoriaVLDB14}).

We start by evaluating some of the properties of the EMD distance with respect
to microaggregation. To minimize EMD between 
the distributions of the confidential attribute within a cluster 
and in the entire data set,
the values of the confidential attribute in the cluster must be as
spread as possible over the entire data set. Consider the case of
a cluster with $k$ records. The following proposition gives 
a lower bound of EMD for such a cluster.
\begin{prop}
	\label{prop:min}Let $T$ be a data set with $n$ records, $A$
	be a confidential attribute of $T$ whose values can be ranked 
and $C$ be a cluster of size $k$.
	The earth mover's distance between $C$ and $T$ 
with respect to  
	attribute $A$ satisfies 
$EMD_{A}(C,T)\ge(n+k)(n-k)/(4n(n-1)k)$.
 If $k$ divides $n$, this lower bound is tight.\end{prop}

\begin{proof}
	The EMD can intuitively be seen as the amount of work needed to transform
	the distribution of attribute $A$ within $C$ into 
the distribution of $A$ over $T$. The ``amount
	of work'' includes two factors: (i) the amount of probability mass
	that needs to be moved and (ii) the distance of the movement. When
	computing EMD for $t$-closeness, the distance of the movements
	of probability mass for numerical attributes is measured as the 
{\em ordered distance}~\cite{Li2007t-Closeness},
	that is, the difference between the ranks of the values of $A$ in $T$
	divided by $n-1$.
	
	For the sake of simplicity, assume that $k$ divides $n$. If that
	is not the case, the distance will be slightly greater, so the lower
	bound we compute is still valid. The probability mass of each of the
	values of $A$ is constant and equal to $1/n$ in $T$, and it 
is  constant and equal to $1/k$ in $C$. This means that the first
	factor that determines the EMD (the amount of probability mass
to be moved) is fixed. Therefore, to minimize 
	EMD we must minimize the second factor (the distance by which the
	probability mass must be moved). Clearly, 
to minimize the distance, the $i$-th value of $A$ in the cluster
	must lie in the middle of the $i$-th group of $n/k$ records of $T$.
	Figure~\ref{fig:min_distance} illustrates this fact.

	\begin{figure}[!t]
		\centering
		\includegraphics[bb=30bp 220bp 800bp 370bp,clip,width=7cm]{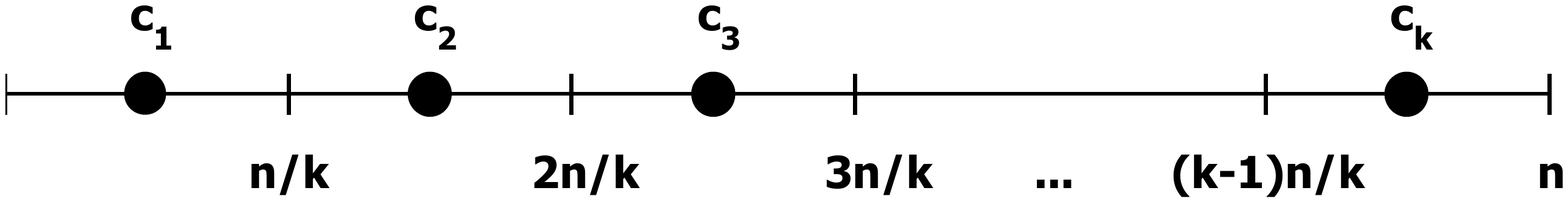}
		\protect\caption{$t$-Closeness first, case $k$ divides $n$. Confidential attribute values $\{c_{1},c_{2},\ldots,c_{k}\}$ of the cluster $C$
			that minimizes the earth mover's distance to $T$. When 
the confidential
			attribute values in $T$ are grouped in $k$ subsets 
of $n/k$ values, $c_i$ is the median of the $i$-th subset for $i=1,\cdots,k$. \label{fig:min_distance}}
	\end{figure}

	In Figure~\ref{fig:min_distance} and using the 
ordered distance, the earth mover's
	distance can be computed as $k$ times the cost of distributing the
	probability mass of element $c_{1}$ among the $n/k$ elements in
	the first subset:
\begin{equation}
\label{minEMD}
	\min(EMD)=k\times\sum_{i=1}^{n/k}\frac{1}{n}\frac{\left|i-\nicefrac{\nicefrac{n}{k}+1}{2}\right|}{n-1}=\frac{(n+k)(n-k)}{4n(n-1)k}
\end{equation}
	Formula (\ref{minEMD}) takes element $(n/k+1)/2$ as the middle element of a
	cluster with $n/k$ elements. Strictly speaking, this is only possible
when $n/k$ is odd.
When $n/k$ is even, we
	ought to take 
either $\lfloor (n/k+1)/2 \rfloor$, the element just before
	the middle, or $\lceil (n/k+1)/2\rceil$, the element just after
	the middle. In any case, the EMD ends up being the same as the 
one obtained in Formula (\ref{minEMD}). 
\end{proof}
{\em Note that, once $n$ and $t$ are fixed, Proposition~\ref{prop:min} 
determines the minimum value of $k$ required
for EMD to be smaller than $t$.} An issue with the construction of 
the $k$ values $c_1$, $\cdots$, $c_k$ 
depicted in Figure~\ref{fig:min_distance}
is that it is too restrictive. For instance, for given values of $n$
and $t$, if the minimal EMD value computed in Proposition~\ref{prop:min}
is exactly equal to $t$, 
then only clusters having as confidential attribute values 
$c_1$, $\cdots$, $c_k$ satisfy $t$-closeness (there may be 
only one such cluster).  
Any other cluster having different confidential attribute
values does not satisfy $t$-closeness. 
Moreover, in the construction of Figure~\ref{fig:min_distance},
the clusters are generated based only on the values of the confidential
attribute, which may lead to a large information loss in
terms of the quasi-identifiers. 

Given the limitations pointed out above, our goal is to guarantee
that the EMD of the clusters is below a specific value but allowing
the clustering algorithm enough freedom to select appropriate records
(in terms of quasi-identifiers) for each of the clusters. The approach
that we propose is similar to the one of Figure~\ref{fig:min_distance}:
we group the records in $T$ into $k$ subsets based on the confidential
attribute and we then generate clusters based on
the quasi-identifiers with the constraint that each 
cluster should contain one record
from each of the $k$ subsets (the specific record is selected based on
the quasi-identifier attributes). Proposition~\ref{prop:upper_bound}
gives an upper bound on the level of $t$-closeness that we attain.
To simplify the derivation and 
the proof, we assume in the proposition that $k$ divides $n$.
\begin{prop}
	\label{prop:upper_bound}Let $T$ be a data set with $n$ records
	and let $A$ be a confidential 
attribute of $T$ whose values can be ranked. 
Let $S=\{S_{1},\ldots,S_{k}\}$
	be a partition of the records in $T$ into $k$ subsets of $n/k$
	records in ascending order of the attribute $A$. Let $C$ be a
	cluster that contains exactly one record from each of the 
subsets $S_{i}$, for $i=1, \cdots, k$.
	Then $EMD(C,T)\le(n-k)/(2(n-1)k)$.\end{prop}
\begin{proof}
	The factors that determine EMD are: (i) the amount of probability
	mass that needs to be moved and (ii) the distance by which it is
moved. The first factor
	is fixed and cannot be modified: each of the records in $T$ has probability
	mass $1/n$, and each of the records in $C$ has probability mass
	of $1/k$. As to the second factor, to find an upper bound to 
	EMD, we need to consider a cluster $C$ that maximizes EMD: the
	records selected for inclusion 
into $C$ must be at the lower (or upper) end 
	of the sets $S_{i}$ for $i=1, \cdots, k$. 
This is depicted in Figure~\ref{fig:max_distance}.
(Note the analogy with the proof of Proposition~\ref{prop:min}: there we 
took the median of each $S_i$ to minimize EMD.)
	
	\begin{figure}[!t]
		\begin{centering}
			\includegraphics[bb=30bp 220bp 800bp 370bp,clip,width=7cm]{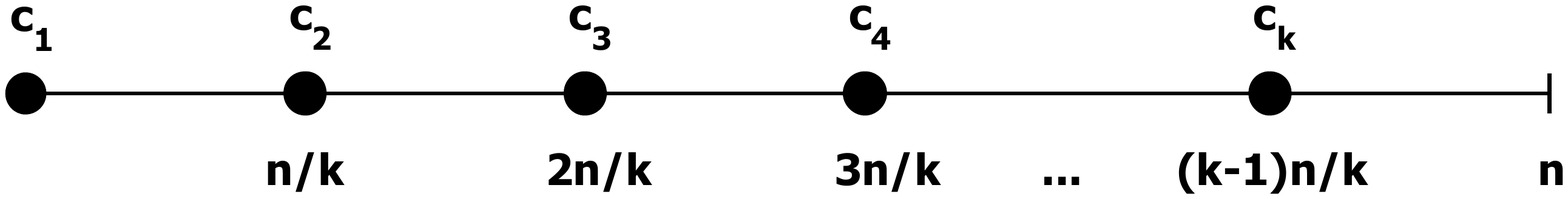}
			
			\par\end{centering}
		\hfil
		\protect\caption{$t$-Closeness first, case $k$ divides $n$. Confidential attribute values $\{c_{1},c_{2},\ldots,c_{k}\}$
			of the cluster $C$ that maximizes the earth mover's distance to $T$. 
When the confidential
			attribute values in $T$ are grouped in $k$ subsets of $n/k$ values, $c_i$ is taken as the minimum value of the $i$-th subset for 
$i=1,\cdots,k$.  \label{fig:max_distance}}
	\end{figure}
	
	EMD for the case in Figure~\ref{fig:max_distance} can be computed
	as $k$ times the cost of distributing the probability mass of $c_{1}$
	among the $n/k$ elements of $S_{1}$:
\begin{equation}
\label{maxEMD}
	\max(EMD)=k\times\sum_{i=1}^{n/k}\frac{1}{n}\frac{i-1}{n-1}=\frac{n-k}{2(n-1)k}
\end{equation}
\end{proof}

With the upper bound on EMD given by Proposition~\ref{prop:upper_bound},
we can determine the cluster size required in the microaggregation:
just replace $EMD(C,T)$ by $t$ on the left-hand side of the bound
and solve for $k$ to get a lower bound for $k$.
For a data set containing $n$ records and for a required level of
$t$-closeness and $k$-anonymity, the cluster size must be
\begin{equation}
	\max\{k,\lceil\frac{n}{2(n-1)t+1}\rceil\}\label{eq:cluster_size}
\end{equation}

To keep things simple, so far we have assumed that $k$ divides
$n$. However, the algorithm to generate $t$-close data sets must
work even if that is not the case. If discarding some records from
the original data set is a viable option, we could discard records
until $k$ divides the new $n$, and proceed as described above. If
records cannot be discarded, some of the clusters would need to contain
more than $k$ records. In particular, we may allow some clusters to have
either $k$ or $k+1$ records.

If we group the records into $k$ sets with $\lfloor n/k \rfloor$ records, then
$r=n \bmod k$ records remain. We propose to assign the remaining
$r$ records to one of the subsets. Then, when generating the clusters, two
records from this subset are added to the first $r$ clusters. This
is only possible if 
$r \leq \lfloor n/k \rfloor$ (the number of remaining records is
not greater than the number of generated clusters); otherwise, there will be
records not assigned to any cluster. Note, however, that
using a cluster size $k$ with $r \geq \lfloor n/k \rfloor$ 
makes no sense: 
since all clusters receive more than $k$ records, what is reasonable
is to adapt to reality by increasing $k$.
Specifically, to avoid having $r \geq \lfloor n/k \rfloor$, $k$ is adjusted as
\begin{equation}
\label{adjust}
k=k+\lfloor (n \bmod k)/\lfloor n/k \rfloor \rfloor.
\end{equation}

Adding two records from one of the subsets to a cluster increases
the EMD of the cluster. To minimize the impact
over the EMD, we need to reduce the work required to distribute the
probability mass of the extra record across the whole range of values.
Hence, the extra record must be close to the median record of the
data set. Figure~\ref{fig:types_of_clusters_1} illustrates the types
of clusters that we allow when $k$ is odd (there is a single subset
in the middle), and Figure~\ref{fig:types_of_clusters_2} illustrates
the types of clusters that we allow when $k$ is even (there are two
subsets in the middle). Essentially, when $k$ is odd, the additional
records are added to $S_{(k+1)/2}$ (the subset in the middle); then, we
generate clusters with size $k$ and clusters with size $k+1$, which
take two records from $S_{(k+1)/2}$. When $k$ is even, the additional
records are split between $S_{(k-1)/2}$ and $S_{(k+1)/2}$ (the
subsets in the middle); then, we generate clusters with size $k$ and clusters
with size $k+1$, some with an additional record from $S_{(k-1)/2}$
and some from $S_{k/2}$.

\begin{figure}[!t]
	\begin{centering}
		\includegraphics[width=8.8cm]{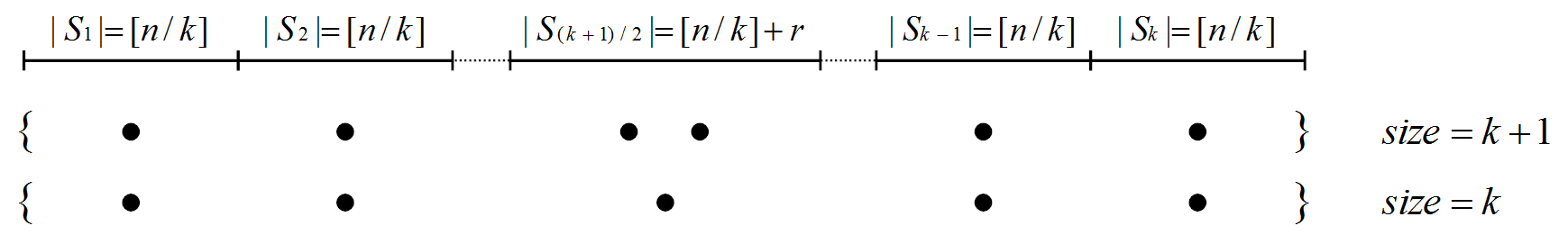}
		
		\par\end{centering}
	\protect\caption{$t$-Closeness first, case $k$ does not divide $n$. Types of clusters when $k$ is odd. Top row, the data set is split
		into $k$ subsets. Central row, cluster with $k+1$ records. 
Bottom row, cluster with $k$ records.\label{fig:types_of_clusters_1}}

\end{figure}

\begin{figure}[!t]
	\begin{centering}
		\includegraphics[width=8.8cm]{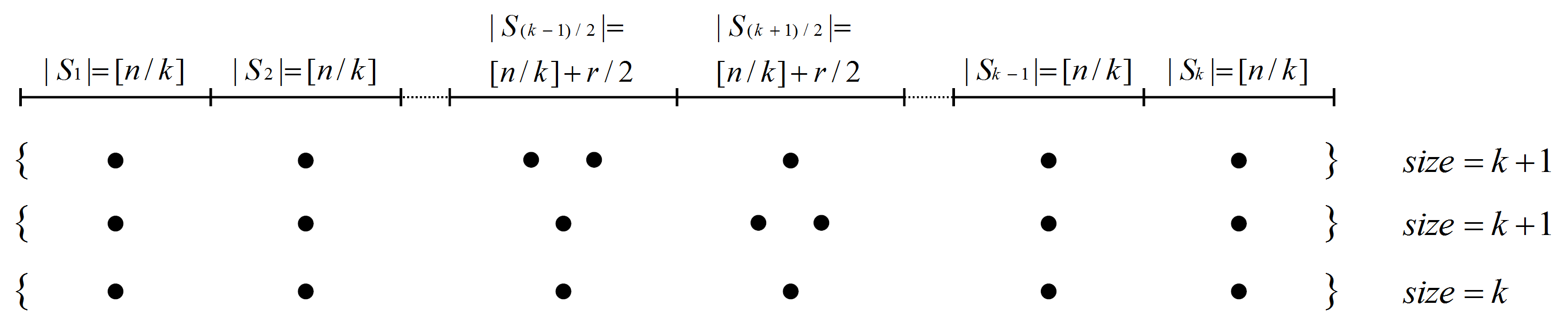}
		
		\par\end{centering}
	\hfil
\protect\caption{$t$-Closeness first, case $k$ does not divide $n$. Types of clusters when $k$ is even. Top row, the data set is split
		into $k$ subsets. Central rows, clusters with $k+1$ records
		(one with two records from $S_{(k-1)/2}$ and the other with two records
		from $S_{(k+1)/2}$). Bottom row, cluster with $k$ records.\label{fig:types_of_clusters_2}}
\end{figure}

Just as we did in Proposition~\ref{prop:upper_bound},
we can compute an upper bound for the EMD of the clusters depicted
in Figures~\ref{fig:types_of_clusters_1} and~\ref{fig:types_of_clusters_2}.
The EMD of a cluster $C$ measures the cost of transforming the distribution
of $C$ to the distribution of the data set. The cost
of the probability mass redistribution can be computed in two steps as 
follows.
First, we want the weight of each subset $S_1,\ldots,S_k$ in cluster
$C$ (the proportion of records in $C$ coming from each subset) to
be equal to the weight of the subset in the data set; 
to this end, we redistribute
the probability mass of the cluster between subsets. 
This redistribution cost, $cost_0$, equals
the EMD between the cluster and the data set when the 
distributions have been discretized to the subsets. 
Then, for each subset $S_i \in \{S_1,\ldots,S_k\}$, 
we compute $cost_i$, an 
upper bound of the cost of distributing the probability
mass $|S_i|/n$ assigned to the subset among its elements
(this is analogous to the mass distribution 
in the proof of Proposition~\ref{prop:upper_bound}). 
The EMD is the sum $cost_0 + cost_1 + \ldots + cost_k$. 
The fact that there are subsets with different sizes and there
are clusters with different sizes makes formulas 
quite tedious and unwieldy, even though the resulting bounds
on EMD are very similar to the one obtained in
Proposition~\ref{prop:upper_bound}. For these reasons, 
we will use the latter as an approximation even when $k$
does not divide $n$; in particular, 
we will determine the cluster size using Expression (\ref{eq:cluster_size}).

Algorithm~\ref{alg:t-closeness-first} formalizes 
the above described procedure to generate a $k$-anonymous $t$-close
data set.
It makes use of Expressions (\ref{eq:cluster_size}) and (\ref{adjust})
 to determine and adjust the cluster size, respectively.
 
 \begin{algorithm}
 	\protect\caption{\label{alg:t-closeness-first}$t$-Closeness-first microaggregation algorithm. 
 		Distances between records are computed
 		in terms of the quasi-identifiers.}
 	
 	\begin{algorithmic}[0]
 		
 	\State {\bf Data:} $X$: original data set
 	\State \hspace{1.11cm}$n$: size of $X$
 	\State \hspace{1.11cm}$k$: minimum cluster size
 	\State \hspace{1.11cm}$t$: $t$-closeness level
 	\State {\bf Result} Set of clusters satisfying $k$-anonymity and $t$-closeness
 	\vspace{0.3cm}
 	\State $k=\max\{k,\lceil \frac{n}{2(n-1)t+1}\rceil\}$
 	\State $k=k + \lceil (n \bmod k)/\lfloor n/k\rfloor \rceil$
 	\State $Clusters=\emptyset$
 	
 	\State Split $X$ into $S_1,\ldots, S_{k}$ subsets with 
 	$\lfloor n/k \rfloor$ records 
 	in ascending order of the confidential attribute, with any remaining 
 	$(n \bmod k)$ records 
 	assigned to the central subset(s)
 	\While{$X\ne\emptyset$}
 		\State $x_a$ = average record of $X$
 		
 		\State $x_0$ = most distant record from $x_a$ in $X$
 		\State $C=\emptyset$
 		\For{$i=1,\ldots,k$}
 			\State $x$ = closest record to $x_0$ in $S_i$
 			\State $C = C \cup \{x\}$
 			\State $S_i=S_i \setminus \{x\}$
 			\State $X = X \setminus \{x\}$
 			\LineComment Take second record from $S_i$ if it contains extra records and no extra record has been already added to $C$
 			\If{$|S_i|>|S_1|$ {\bf and} $|C|=i$}
 				\State $x$ = closest record to $x_0$ in $S_i$
 				\State $C = C \cup \{x\}$
 				\State $S_i=S_i \setminus \{x\}$
 				\State $X = X \setminus \{x\}$
 			\EndIf
 		\EndFor
 		\State $Clusters = Clusters \cup \{C\}$
 		
 		\If{$X\ne\emptyset$}
 			\State $x_1$ = most distant record from $x_0$ in $X$
 			\State $C=\emptyset$
 			\For{$i=1,\ldots,k$}
 				\State $x$ = closest record to $x_1$ in $S_i$
 				\State $C = C \cup \{x\}$
 				\State $S_i=S_i \setminus \{x\}$
 				\State $X = X \setminus \{x\}$
 				\If{$|S_i|>|S_1|$ {\bf and} $|C|=i$}
 					\State $x$ = closest record to $x_1$ in $S_i$
 					\State $C = C \cup \{x\}$
 					\State $S_i=S_i \setminus \{x\}$
 					\State $X = X \setminus \{x\}$					
 				\EndIf	
 			\EndFor
 			\State $Clusters = Clusters \cup \{C\}$
 		\EndIf
 	\EndWhile
 	\State {\bf return} $Clusters$
 	\end{algorithmic}
 	
 \end{algorithm}
 
 In terms of computational cost, Algorithm~\ref{alg:t-closeness-first} has a great advantage over 
 Algorithms~\ref{alg:micro_merge} and~\ref{alg:k-anonymity-first}: when running
 Algorithm~\ref{alg:t-closeness-first}, we know that by
construction the generated clusters satisfy $t$-closeness, so there is no need to compute any EMD distance. Algorithm~\ref{alg:t-closeness-first}
 has cost $\mathcal{O}(n^2/k)$, the same cost order as MDAV (on which 
it is based). Actually,
 Algorithm~\ref{alg:t-closeness-first} is even slightly more efficient 
than MDAV: all operations being
 equal, some of the computations that MDAV performs on the entire data set are performed by 
 Algorithm~\ref{alg:t-closeness-first} just on one of the 
subsets of $n/k$ records.

\section{Empirical evaluation}
\label{sec:empirical}

In this section we empirically evaluate and compare
the proposed algorithms using several data sets and according
to different metrics: actual cluster size, 
speed and scalability, and data utility preservation. 

\subsection{Actual cluster size}
\label{sec:behavior}

In a first battery of tests we used as evaluation data 
the Census data set~\cite{Brand}, which is usual to test 
privacy protection methods~\cite{Yanc02,Lasz05,Domingo10}
and contains 1,080 records with numerical attributes. 
Similar to~\cite{Domingo10}, we took attributes
TAXINC (Taxable income amount) and POTHVAL (Total other persons income)
as quasi-identifiers, and FEDTAX (Federal income tax liability)
and FICA (Social Security retirement payroll deduction)
as confidential attributes. 

Because $k$-anonymity and $t$-closeness pursue different goals
(the former clusters records with similar quasi-identifiers while 
the latter requires clusters with a distribution of 
confidential attributes similar to the one of the entire data set),
we defined two data sets according to the
correlation between the values of quasi-identifier and confidential attributes:

\begin{itemize}
\item {\em Moderately correlated data set (MCD)}. It consists of 
1,080 records with TAXINC and POTHVAL 
as quasi-identifier attributes, and FEDTAX as confidential attribute. The correlation 
between both types of attributes is 0.52. This represents the most 
usual scenario in which quasi-identifiers and confidential attributes 
show some correlation. 
\item {\em Highly correlated data set (HCD)}. 
It uses the same quasi-identifiers as MCD,
but it takes FICA as confidential attribute. The correlation between
both types of attributes is 0.92. This highly correlated data set represents
a worst-case scenario for our algorithms because, to fulfill a certain 
$t$-closeness level 
({\em i.e.}, to ensure a certain distribution of confidential values), 
we are likely to be forced to microaggregate records
with significantly diverse quasi-identifier values, thereby 
incurring more information loss than in the MCD data set.
\end{itemize}

By applying the three algorithms to these two data sets
for different values of $k$ and $t$, we will show 
how close to $k$ are the sizes of clusters formed by 
each algorithm for each value of $t$ to be enforced.
{\em To minimize information loss, the closer 
all cluster sizes to $k$, the better.}
The $k$ values have been taken in the range 2-30, 
which covers the most usual $k$-anonymity values ({\em e.g.} 
$k$ is taken between 3 and 10 in~\cite{Domingo01}),
whereas the $t$ values have been taken in the range 0.01-0.25 
(where 0.25 is the upper bound of $t$-closeness for this data set for 
the lowest $k$, that is $k=2$, 
according to Proposition~\ref{prop:upper_bound}).

We start by analyzing the behavior of Algorithm~\ref{alg:micro_merge}, 
in which records are first microaggregated in clusters of size $k$ that 
are thereafter merged until $t$-closeness is fulfilled. Table~\ref{tab:micro-merge} shows the
actual level of microaggregation that results from the merging process: 
{\em minimum}, that is, the size of the smallest cluster 
(which determines the actual $k$-anonymity level achieved), 
and {\em average},
that is, the average size of the merged clusters.

\begin{table*}
\caption{\label{tab:micro-merge}Algorithm~\ref{alg:micro_merge}: actual microaggregation (minimum and average
size of the clusters, respectively) resulting for several values of $k$ and $t$ for the MCD and HCD data sets}
\centering
\scriptsize
\begin{tabular}{|l|cc|cc|cc|cc|cc|cc|cc|}\hline
 & \multicolumn{2}{c|}{$t=0.01$} & \multicolumn{2}{c|}{$t=0.05$} & \multicolumn{2}{c|}{$t=0.09$} & \multicolumn{2}{c|}{$t=0.13$} & \multicolumn{2}{c|}{$t=0.17$} & \multicolumn{2}{c|}{$t=0.21$} & \multicolumn{2}{c|}{$t=0.25$}\\
& MCD & HCD & MCD & HCD & MCD & HCD & MCD & HCD & MCD & HCD & MCD & HCD & MCD & HCD\\
\hline
$k=2$ & 1080/1080 & 1080/1080  & 56/120 & 36/98  & 20/42 & 16/31 & 8/20 & 8/52 & 4/10 & 4/9 & 4/7 & 4/7 & 2/8 & 2/5 \\
$k=5$ & 1080/1080 & 1080/1080  & 385/540 & 200/216  & 40/154 & 40/60 & 20/47 & 20/80 & 10/24 & 10/21 & 10/17 & 10/15 & 5/12 & 5/11 \\
$k=10$ & 1080/1080 & 1080/1080  & 1080/1080 & 1080/1080  & 110/270 & 180/216 & 40/108 & 40/190 & 20/57 & 20/47 & 20/35 & 20/31 & 10/24 & 10/20 \\
$k=15$ & 1080/1080 & 1080/1080  & 1080/1080 & 495/540  & 135/360 & 195/270 & 45/90 & 60/270 & 30/64 & 30/68 & 30/45 & 30/54 & 15/33 & 15/33 \\
$k=20$ & 1080/1080 & 1080/1080  & 380/540 & 240/360  & 160/216 & 180/216 & 80/154 & 60/140 & 40/83 & 40/72 & 40/54 & 40/60 & 20/37 & 20/40 \\
$k=25$ & 1080/1080 & 1080/1080  & 1080/1080 & 1080/1080  & 1080/1080 & 230/360 & 455/540 & 50/250 & 50/180 & 50/98 & 50/90 & 50/72 & 25/72 & 25/48 \\
$k=30$ & 1080/1080 & 1080/1080  & 540/540 & 1080/1080  & 270/360 & 330/540 & 120/180 & 150/390 & 60/98 & 60/108 & 60/77 & 60/90 & 30/57 & 30/57 \\
\hline
\end{tabular}
\end{table*} 

It can be seen that, in many cases, the actual level of microaggregation
is significantly higher than the value of $k$. This is undesirable because 
the larger the clusters, the higher the information loss. 
We also see that {\em the size of the clusters tends to increase} for both data sets {\em as}:
\begin {itemize}
\item i) {\em the parameter $t$ of  $t$-closeness decreases}: since clusters have been created without
considering the desired $t$-closeness, it is unlikely that 
they satisfy it as $t$ gets smaller. 
Thus, to decrease, if necessary, the distance between
the distribution of confidential attributes within each cluster and 
over the entire data set,
the algorithm merges the already created clusters (thereby 
increasing their cardinality);
in the worst case ({\em i.e.}, $t$ around 0.01-0.05), this implies grouping all 1,080 records in a single cluster.  
\item ii) {\em the initial level of $k$-anonymity increases}: the coarser the initial microaggregation,  
the more effort ({\em i.e.}, merging) is needed to achieve a certain $t$-closeness level. 
\end {itemize}

We also observe a noticeable difference between the minimum and average cardinality of the
clusters, which suggests that the microggregation of records that we obtain in practice
is far from optimal.

\begin{table*}
\caption{\label{tab:k-anonymity-first}Algorithm~\ref{alg:k-anonymity-first}: actual microaggregation (minimum and average
size of the clusters, respectively) resulting for several values of $k$ and $t$ for the MCD and HCD data sets.}
\centering
\scriptsize
\begin{tabular}{|l|cc|cc|cc|cc|cc|cc|cc|}\hline
 & \multicolumn{2}{c|}{$t=0.01$} & \multicolumn{2}{c|}{$t=0.05$} & \multicolumn{2}{c|}{$t=0.09$} & \multicolumn{2}{c|}{$t=0.13$} & \multicolumn{2}{c|}{$t=0.17$} & \multicolumn{2}{c|}{$t=0.21$} & \multicolumn{2}{c|}{$t=0.25$}\\
& MCD & HCD & MCD & HCD & MCD & HCD & MCD & HCD & MCD & HCD & MCD & HCD & MCD & HCD\\
\hline
$k=2$ & 164/216 & 156/360  & 8/10 & 8/11  & 4/7 & 4/7 & 4/6 & 4/4 & 2/3 & 2/3 & 2/3 & 2/3 & 2/3 & 2/3 \\
$k=5$ & 40/64 & 80/154  & 10/16 & 10/10  & 5/7 & 5/8 & 5/7 & 5/8 & 5/7 & 5/7 & 5/7 & 5/8 & 5/6 & 5/7 \\
$k=10$ & 40/108 & 80/135  & 10/17 & 10/17  & 10/17 & 10/17 & 10/15 & 10/16 & 10/15 & 10/14 & 10/13 & 10/14 & 10/12 & 10/12 \\
$k=15$ & 30/57 & 30/60  & 15/28 & 15/30  & 15/25 & 15/26 & 15/23 & 15/23 & 15/23 & 15/22 & 15/19 & 15/21 & 15/16 & 15/17 \\
$k=20$ & 40/54 & 40/49  & 20/37 & 20/43  & 20/35 & 20/36 & 20/32 & 20/32 & 20/31 & 20/29 & 20/26 & 20/28 & 20/22 & 20/23 \\
$k=25$ & 50/51 & 50/51  & 25/51 & 25/51  & 25/43 & 25/43 & 25/39 & 25/39 & 25/40 & 25/37 & 25/32 & 25/35 & 25/28 & 25/26 \\
$k=30$ & 30/57 & 60/64  & 30/68 & 30/64  & 30/54 & 30/54 & 30/49 & 30/47 & 30/47 & 30/43 & 30/37 & 30/42 & 30/34 & 30/34 \\
\hline
\end{tabular}
\end{table*} 

Table~\ref{tab:k-anonymity-first} shows the results for Algorithm~\ref{alg:k-anonymity-first}.
With this algorithm, we observe that the actual microaggregation levels are significantly smaller
than in the previous case for the same values of $k$ and $t$, and so is the difference
between the minimum and average cardinality of the clusters. 
Now $t$-closeness is enforced after creating 
each cluster rather than after creating all clusters. Thus, once a cluster is created, 
some of the $k$ records in that cluster may be replaced
by unclustered records until $t$-closeness is satisfied; doing so 
does not increase the cardinality of the cluster, even though it 
may end up clustering records 
with less homogeneous quasi-identifiers and thereby yielding 
a higher loss of information. Only 
if the replacement does not satisfy the desired $t$-closeness, the clusters 
are merged
like in Algorithm~\ref{alg:micro_merge}, 
thereby increasing the microaggregation level
(in fact, as suggested in Section~\ref{sub:k-anonymity-first}, we use 
Algorithm~\ref{alg:k-anonymity-first} as the microaggregation
function of Algorithm~\ref{alg:micro_merge}). 
The results shown in Table~\ref{tab:k-anonymity-first}
suggest that this process occurs for the smallest $t$-closeness values
({\em i.e.}, 0.01-0.05), which are the ones that impose the strictest constraint. 

The differences between the two data sets are more noticeable if we look at the
average cardinality of the clusters: the HCD data set results in a larger average cardinality,
because the initial clusters present more homogeneous 
confidential values (these are very correlated to the more homogeneous
quasi-identifier values obtained for the first clusters) 
and tend to require more effort ({\em i.e.}, replacements and mergers) 
to attain $t$-closeness. 

\begin{table*}
\caption{\label{tab:t-closeness-first}Algorithm~\ref{alg:t-closeness-first}: actual microaggregation (minimum and average
size of the clusters, respectively) resulting for several values of $k$ and $t$ for the MCD and HCD data sets.}
\centering
\scriptsize
\begin{tabular}{|l|cc|cc|cc|cc|cc|cc|cc|}\hline
 & \multicolumn{2}{c|}{$t=0.01$} & \multicolumn{2}{c|}{$t=0.05$} & \multicolumn{2}{c|}{$t=0.09$} & \multicolumn{2}{c|}{$t=0.13$} & \multicolumn{2}{c|}{$t=0.17$} & \multicolumn{2}{c|}{$t=0.21$} & \multicolumn{2}{c|}{$t=0.25$}\\
& MCD & HCD & MCD & HCD & MCD & HCD & MCD & HCD & MCD & HCD & MCD & HCD & MCD & HCD\\
\hline
$k=2$ & 49/49 & 49/49  & 10/10 & 10/10  & 6/6 & 6/6 & 4/4 & 4/4 & 3/3 & 3/3 & 3/3 & 3/3 & 2/2 & 2/2 \\
$k=5$ & 49/49 & 49/49  & 10/10 & 10/10  & 6/6 & 6/6 & 5/5 & 5/5 & 5/5 & 5/5 & 5/5 & 5/5 & 5/5 & 5/5 \\
$k=10$ & 49/49 & 49/49  & 10/10 & 10/10  & 10/10 & 10/10 & 10/10 & 10/10 & 10/10 & 10/10 & 10/10 & 10/10 & 10/10 & 10/10 \\
$k=15$ & 49/49 & 49/49  & 15/15 & 15/15  & 15/15 & 15/15 & 15/15 & 15/15 & 15/15 & 15/15 & 15/15 & 15/15 & 15/15 & 15/15 \\
$k=20$ & 49/49 & 49/49  & 20/20 & 20/20  & 20/20 & 20/20 & 20/20 & 20/20 & 20/20 & 20/20 & 20/20 & 20/20 & 20/20 & 20/20 \\
$k=25$ & 49/49 & 49/49  & 25/25 & 25/25  & 25/25 & 25/25 & 25/25 & 25/25 & 25/25 & 25/25 & 25/25 & 25/25 & 25/25 & 25/25 \\
$k=30$ & 49/49 & 49/49  & 30/30 & 30/30  & 30/30 & 30/30 & 30/30 & 30/30 & 30/30 & 30/30 & 30/30 & 30/30 & 30/30 & 30/30 \\
\hline
\end{tabular}
\end{table*}

Finally, Table~\ref{tab:t-closeness-first} shows the results for Algorithm~\ref{alg:t-closeness-first}.
Figures in this table show that Algorithm~\ref{alg:t-closeness-first}
is the one achieving an actual microaggregation level closest to 
the desired $k$.
Moreover, since the cardinality of the data sets (1,080 records) is a multiple of the values of $k$,
all clusters can be formed with the same cardinality $k$ ({\em i.e.}, clusters are perfectly balanced).
Indeed, as stated in Section~\ref{sub:t-closeness-first},
Algorithm~\ref{alg:t-closeness-first} 
seeks the smallest clusters whose cardinality 
is at least $k$ and which satisfy a 
pre-specified level of $t$-closeness. 
To do so, it prioritizes the fulfillment of 
$t$-closeness over the homogeneity of quasi-identifiers 
in cluster formation. Because of this strategy,
there are no differences between the MCD and HCD data sets;
in fact, we can see that for most parameter choices and for both data sets
the minimum and average cluster sizes are $k$.

In comparison with Algorithm~\ref{alg:k-anonymity-first}, we observe that,
even though in some cases ({\em e.g.}, for $t=0.05$ and $k=2$) the minimum
cardinality is greater with Algorithm~\ref{alg:t-closeness-first}, 
the average cardinality is always smaller with this algorithm.
This is a consequence of the more precise 
microaggregation implemented by Algorithm~\ref{alg:t-closeness-first}.

\subsection{Speed and scalability}
\label{perform}

The second part of the evaluation focuses on measuring the speed
and scalability of the three algorithms with a larger data set.

To that end, we took a higher-dimensional data set from the 
the Patient Discharge Data for year 2010 of Californian hospitals,
which are provided by 
California's Office of Statewide Health Planning and Development~\cite{oshpd}. 
We took the data set with the largest number of entries
(Cedars Sinai Medical Center, with 55,668 patient records). 
From these, we removed records with missing attribute values and 
obtained a final data set with 23,435 records. Each record
consists of 7 quasi-identifier attributes ({\em e.g.}, patient's age, zip code, admission date, etc.)
plus one confidential attribute that specifies the amount charged
for the patient's stay in the hospital. The correlation
between the quasi-identifier attributes and the confidential one is just 0.129.

The run time of the three algorithms for the Patient Discharge data set is shown in
Figure~\ref{fig:time} as a function of the value of $t$ to be attained.
We set $k=2$ in order to give maximum freedom to the algorithms in 
adapting the microaggregation to the desired value of $t$ (again between 0.01 and 
the maximum upper bound of 0.25), and force them to create the greatest
number of clusters (which the is worst case from the run time perspective).

\begin{figure}[!t]
	\centering
	\includegraphics[width=8cm]{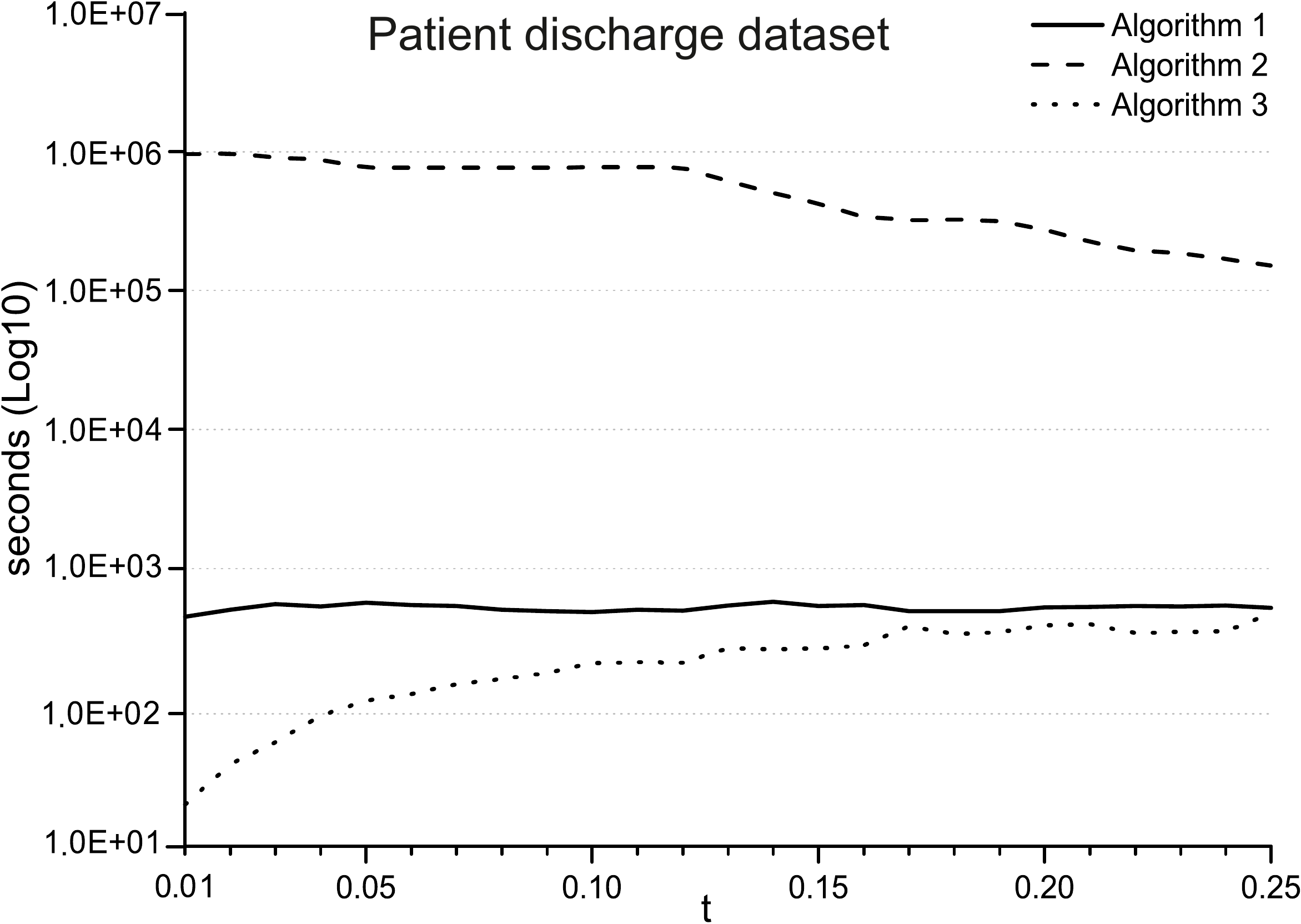}
	\protect\caption{Run time (in seconds with $log_{10}$ scale) for the three algorithms with $k=2$ 
and values of $t$	between 0.02 and 0.25 for the Patient Discharge data set\label{fig:time}}
\end{figure}

Run time figures are coherent with the theoretical analysis of 
computational costs for the three algorithms. Algorithms ~\ref{alg:micro_merge} 
and ~\ref{alg:t-closeness-first} are significantly more efficient than Algorithm~\ref{alg:k-anonymity-first}
(note the logarithmic scale of the Y-axis), because the former have just 
the quadratic 
cost of the underlying microaggregation algorithm, 
whereas the latter has a cubic cost resulting from
the rearrangement of records required to fulfill $t$-closeness after the creation of each cluster.
Indeed, Algorithm~\ref{alg:k-anonymity-first} may 
not scale well for large
data sets, whereas the other two algorithms scale 
as well as the underlying microaggregation.
At a closer look, Algorithm~\ref{alg:t-closeness-first} is 
significantly more efficient than Algorithm~\ref{alg:micro_merge} for low values of $t$.
The reason is that, although the cost of both algorithms is $\mathcal{O}(n^2/k)$,
Algorithm~\ref{alg:t-closeness-first} optimally updates 
the value of $k$ in terms of the actual $t$:
for small values of $t$, the value of $k$ is large (see Equation~\ref{eq:cluster_size}),
which reduces the computational cost.
In contrast, Algorithm~\ref{alg:micro_merge} only takes $t$
into account after the entire microaggregation has been performed. 
Finally, the run time of Algorithm~\ref{alg:k-anonymity-first} tends 
to decrease for large $t$ because, in this case, clusters are more
likely to (nearly) fulfill
$t$-closeness, thus requiring less rearrangement of records 
after each iteration.

\subsection{Data utility preservation}
\label{utility}

So far, the comparison between algorithms 
has been made only in terms of cluster sizes and run time.
Let us now examine to what extent each algorithm preserves 
the data utility for a certain privacy level.
Indeed, the different microaggregation strategies
and the actual levels of microaggregation achieved by the three algorithms
have a direct influence on the utility of the anonymized results.
In the literature, the utility of an anonymized output is evaluated in terms of \textit{information loss}, that is, 
the discrepancies between the original and the anonymized data set.  
The Sum of Squared Errors (SSE) is a well-known information loss measure, 
which is well-suited to capture the impact of creating
equivalence classes by means of $k$-anonymous microaggregation.
SSE is defined as 
the sum of squares of attribute distances between records
in the original data set $X$ and their versions in the anonymized data set.
However, since SSE provides absolute error values, we 
normalized it to obtain a measure that is independent of
the data set size 
(number of records and attributes) and the ranges of attribute values:

\begin{equation}
\label{SSE}
SSE = \frac{1}{n} \sum_{x_{j} \in X} \frac{1}{m} \sum_{a_{j}^{i} \in x_{j}} (NED(a_{j}^{i},(a_{j}^{i})'))^2 
\end{equation}

where $n$ is the number of records, $m$ is the number of attributes,
 $a_{j}^{i}$ is the value of the $i$-th attribute
for the $j$-th original record, $(a_{j}^{i})'$ represents its 
anonymized version and
$NED(\cdot,\cdot)$ corresponds to the Normalized Euclidean Distance.
Notice that with a high SSE, that is, a high information loss, 
a lot of data uses are severely damaged, like
for example subdomain analyses (analyses restricted to parts of the data set).   

To fairly and clearly compare the three algorithms, we first took $k=2$ for $k$-anonymity
with $t$ values between 0.01 and 0.25 for $t$-closeness. 
In this manner, any actual cluster size $k > 1$ is feasible 
and the algorithms have the
greatest freedom to microaggregate records to fulfill the desired $t$-closeness.
SSE values for each value of $t$ are shown in Figure~\ref{fig:SSE} for the three data sets. 

\begin{figure}[!t]
\begin{centering}
	\includegraphics[width=8cm]{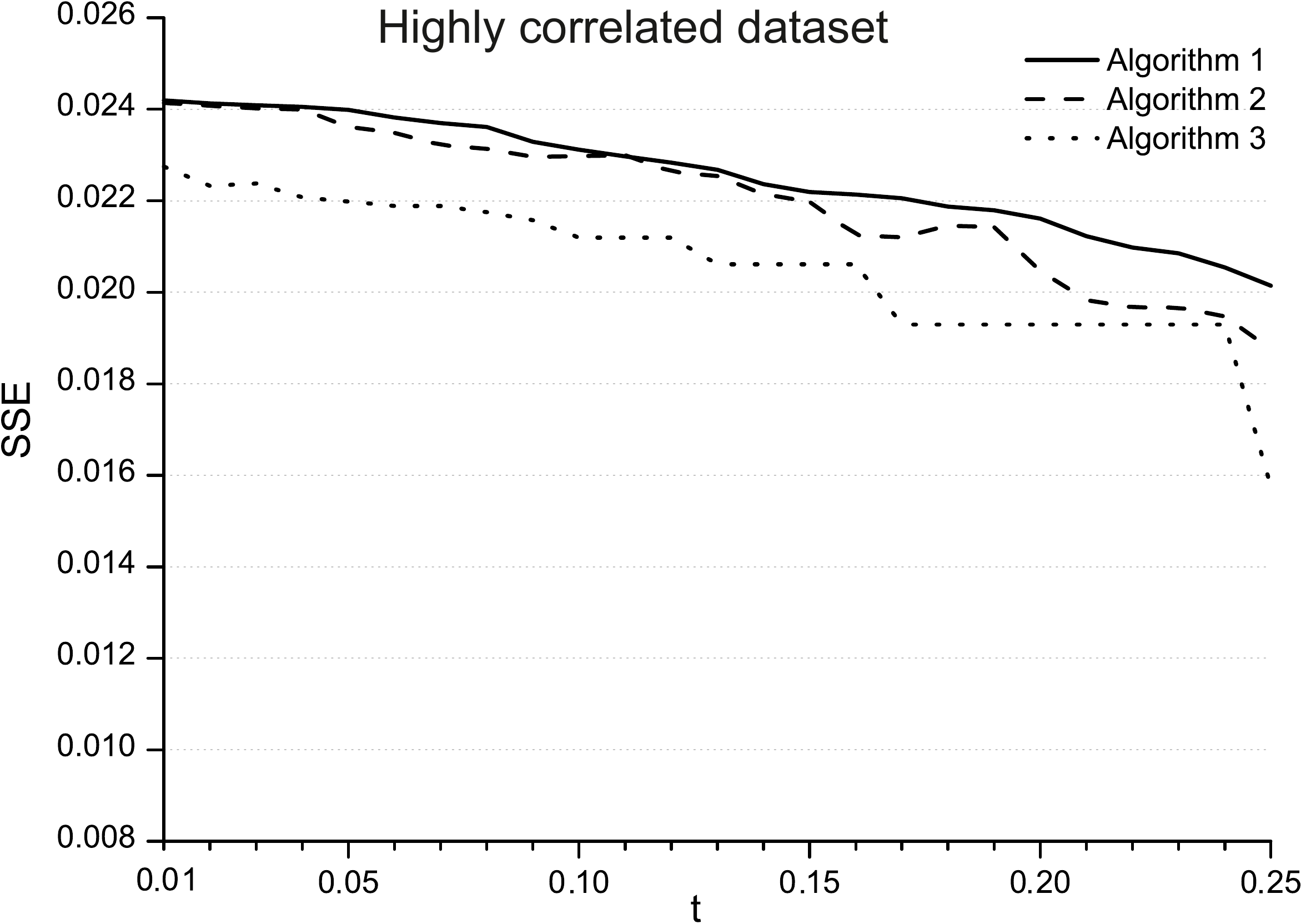}\\
	\vspace{0.2cm}	
\includegraphics[width=8cm]{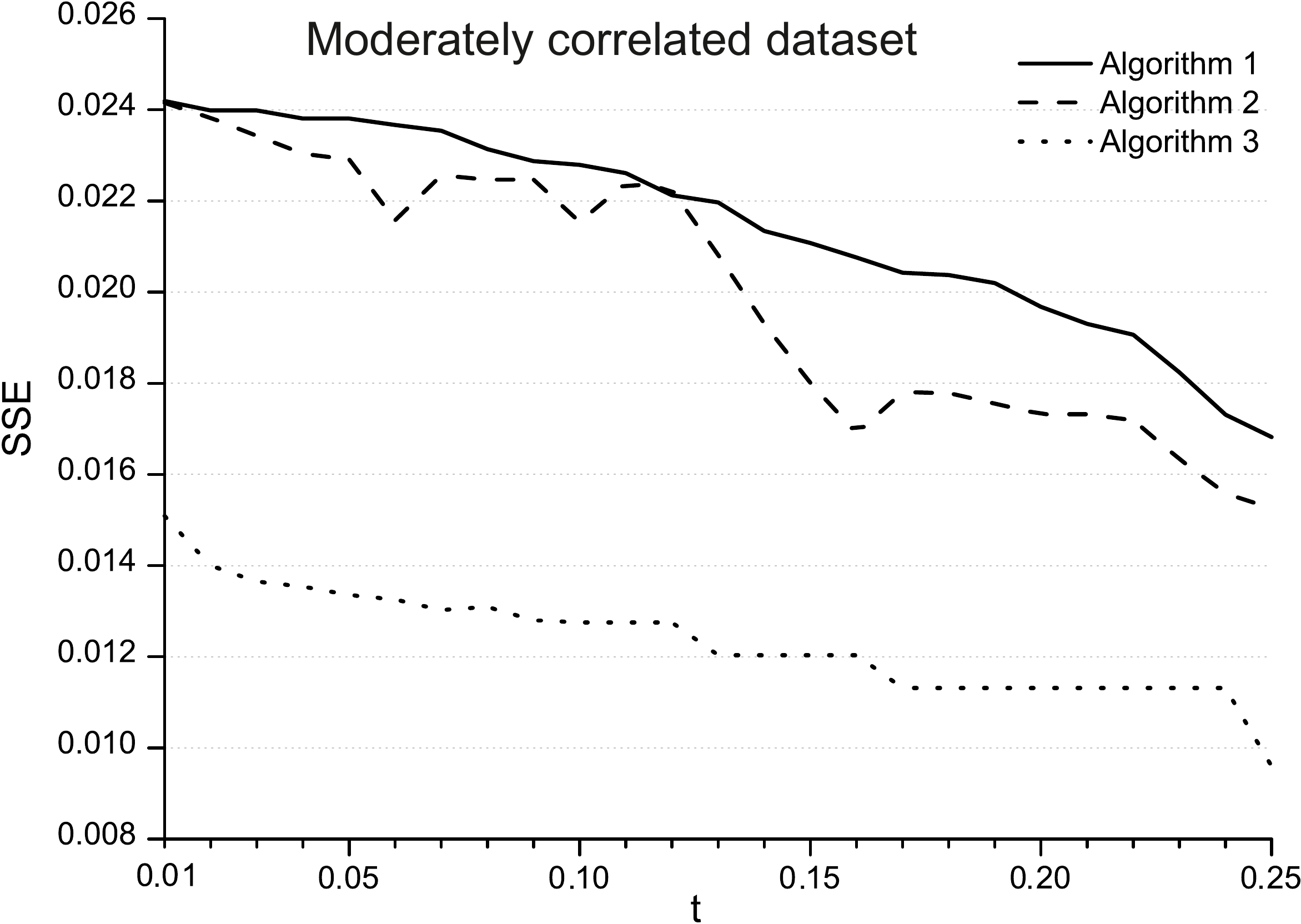}\\
	\vspace{0.2cm}
	\includegraphics[width=8cm]{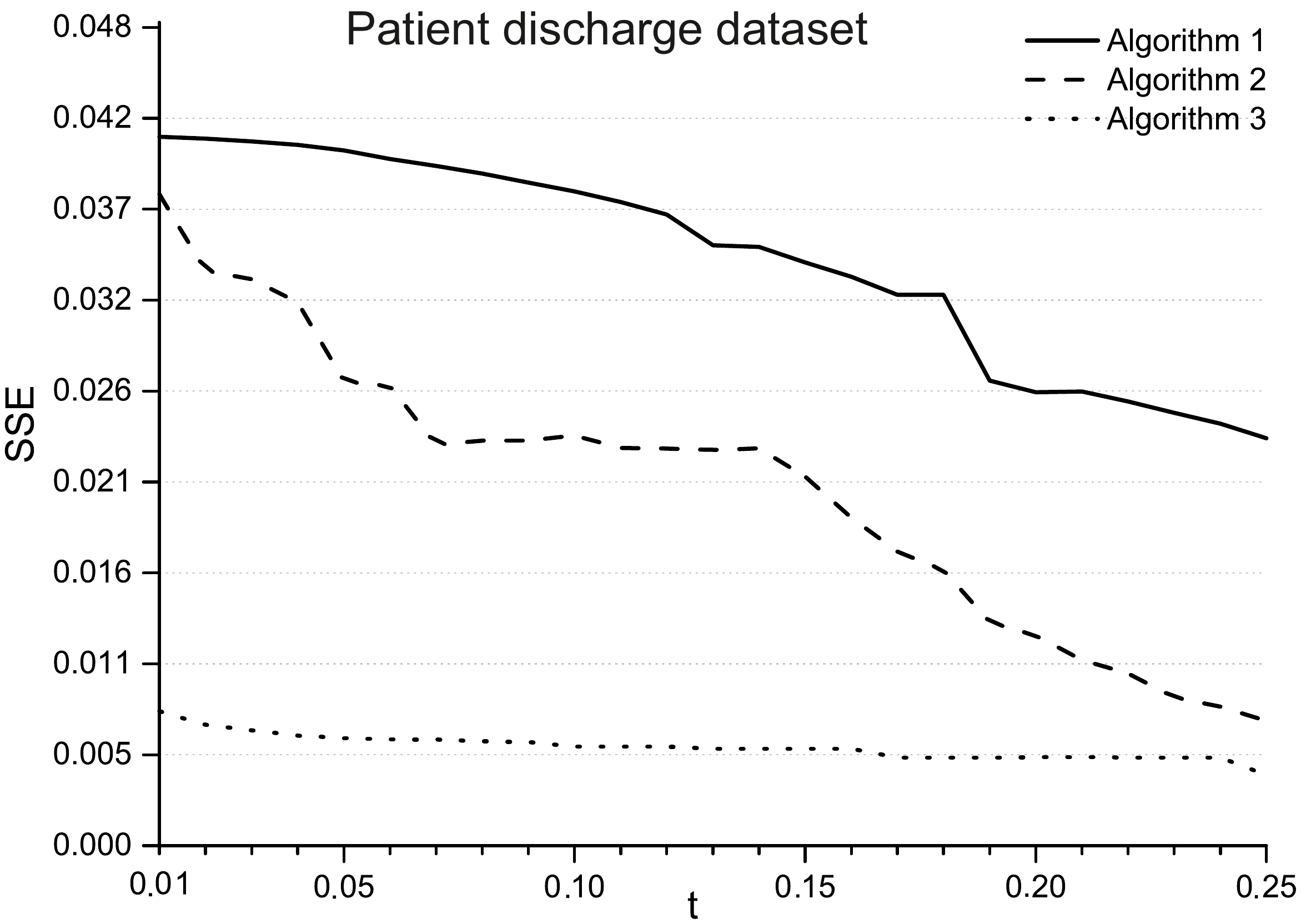}\\
	\vspace{0.2cm}
\par\end{centering}
\protect\caption{Normalized SSE values for the three algorithms with $k=2$ 
and values of $t$ between 0.02 and 0.25 for the HCD (top), MCD (middle) and 
Patient Discharge (bottom) data sets\label{fig:SSE}}
	\end{figure}

All graphs show that Algorithm~\ref{alg:k-anonymity-first} 
improves on Algorithm~\ref{alg:micro_merge} and, in turn, 
Algorithm~\ref{alg:t-closeness-first} improves 
on Algorithm~\ref{alg:k-anonymity-first}.
Thus, we can see that {\em the earlier we consider
the fulfillment of $t$-closeness in the microaggregation step, the
more utility is preserved in the output}. This may seem paradoxical, 
because a $t$-closeness aware microaggregation that prioritizes 
the distribution of confidential values 
(such as the one in Algorithm~\ref{alg:t-closeness-first})
is likely to cluster records with heterogeneous quasi-identifier values, 
and thereby incur higher information loss. 
Some of this is apparent in Figure~\ref{fig:SSE}: 
Algorithm~\ref{alg:t-closeness-first} improves much more on
the other two algorithms for the MCD and Patient Discharge than for the HCD data set,
because cluster homogeneity for HCD is harder to 
reconcile with the $t$-closeness requirement due to the higher
correlation of quasi-identifiers and the confidential attribute.
However, on the other hand, the fact that
the $k$-anonymous microaggregation is aware of the level of $t$-closeness that
should be satisfied also produces smaller clusters (of size closer 
to the desired $k$),
which is beneficial to keep SSE low.
In contrast, the other algorithms, and especially 
Algorithm~\ref{alg:micro_merge}, prioritize
quasi-identifier values in the $k$-anonymous microaggregation and, 
hence, they require a lot of cluster merging and/or manipulation 
to attain $t$-closeness.
This tends to produce larger clusters (as shown by the experiments
 on cluster sizes), 
whose aggregation incurs a greater loss of information, 
which is nonetheless fairly independent of the correlation
between quasi-identifiers and confidential attributes;
this is especially noticeable for the Patient Discharge data set,
in which Algorithm~\ref{alg:micro_merge} behaves significantly
worse than the other two.

To sum up, {\em the increase of information loss that the lower cluster
homogeneity of $t$-closeness aware microaggregation might cause
is more than compensated by the information loss reduction 
resulting from smaller clusters}.

Finally, we also evaluated the evolution of the normalized SSE as a function of 
both $k$ and $t$. As a reference, Figure~\ref{fig:SSE-3D} shows 
this evolution for the three algorithms with the MCD data set.

\begin{figure}[!t]
\begin{centering}
	\includegraphics[width=8cm]{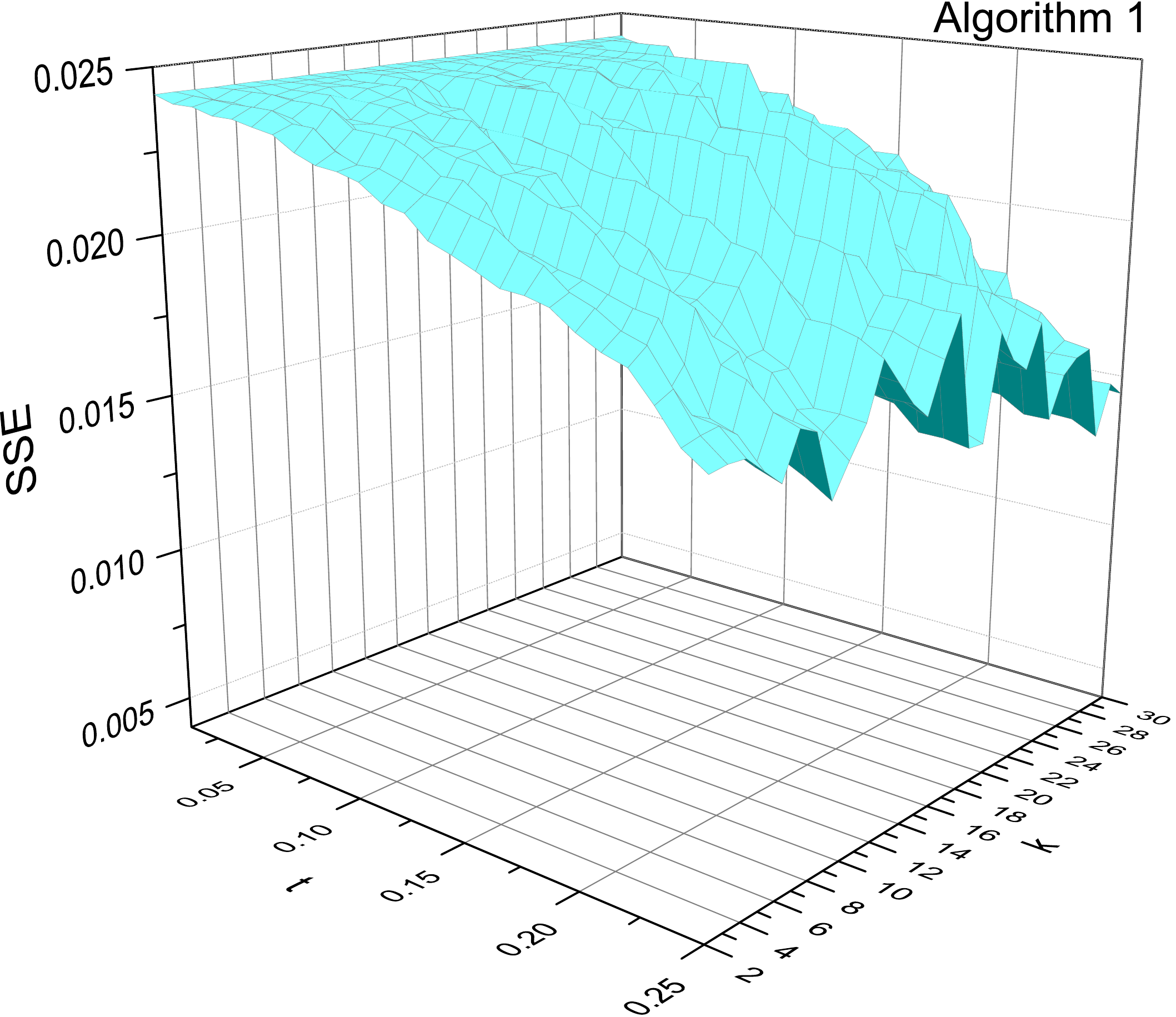}\\
	\vspace{0.4cm}	
\includegraphics[width=8cm]{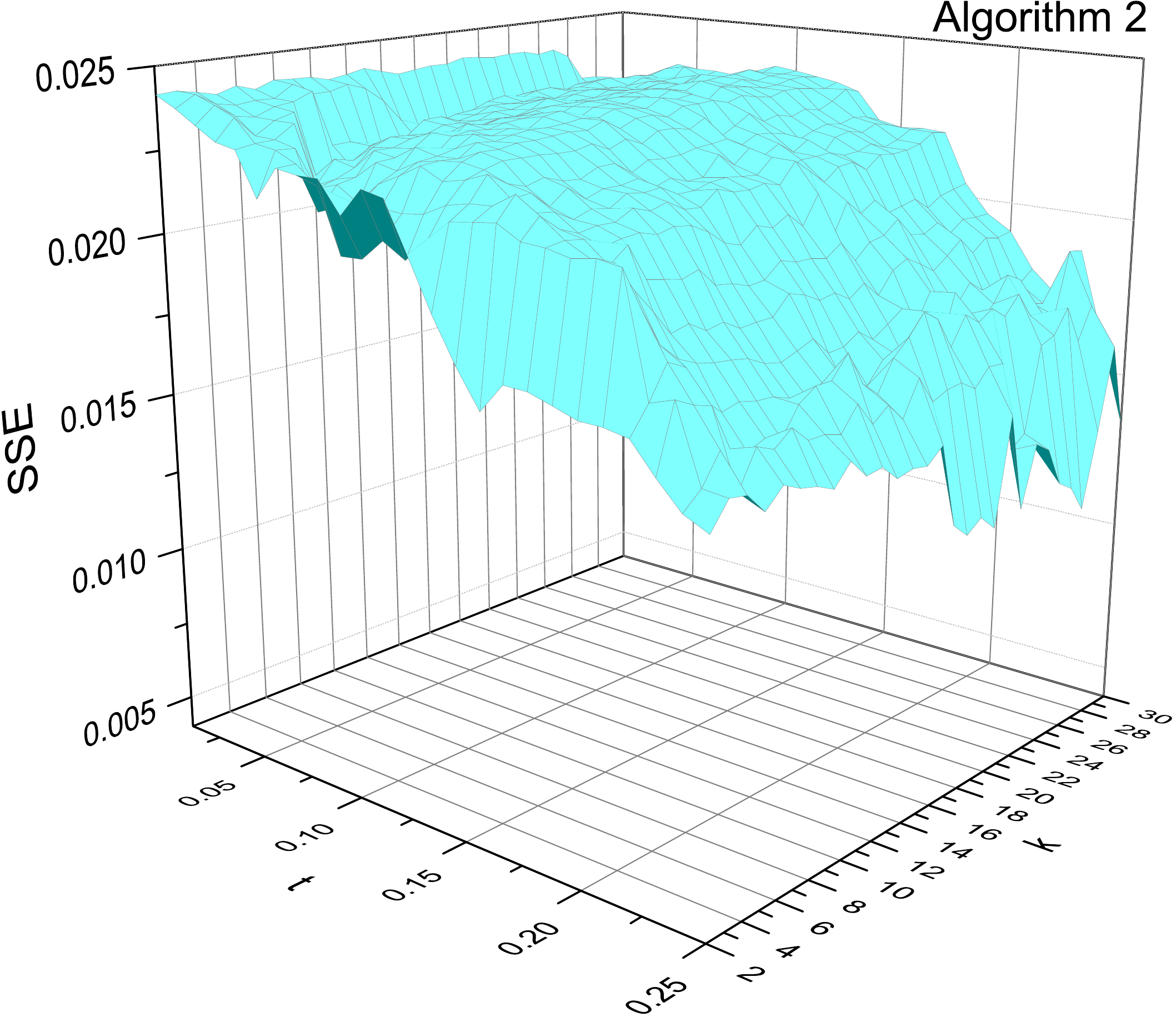}\\
	\vspace{0.4cm}
	\includegraphics[width=8cm]{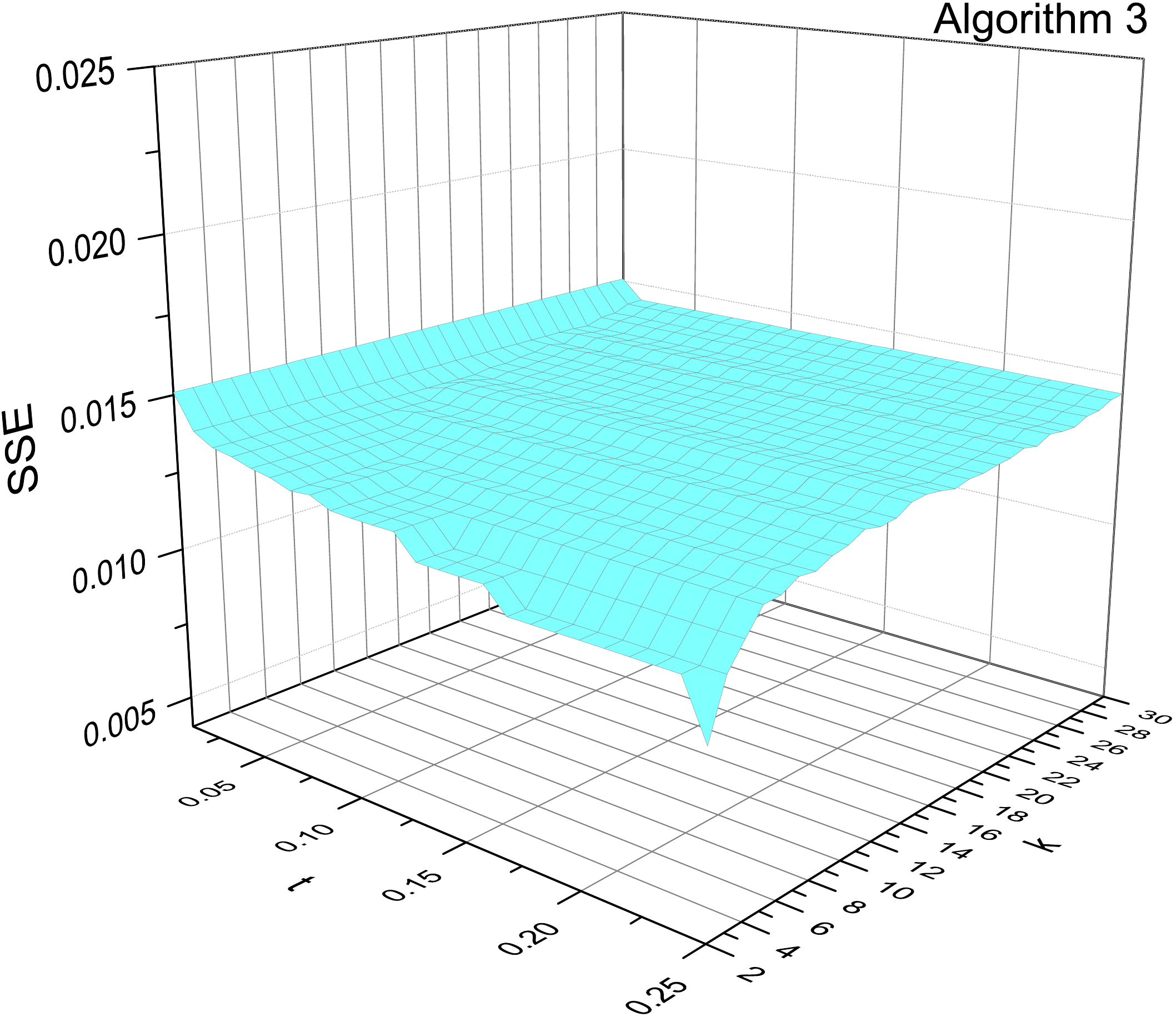}\\
	\vspace{0.4cm}
\par\end{centering}
		\protect\caption{Normalized SSE for the three algorithms 
for $k$ between 2 and 30 
and $t$ between 0.02 and 0.25 for the MCD data set\label{fig:SSE-3D}}
	\end{figure}

First, we can see that some of the advantages of Algorithm~\ref{alg:t-closeness-first}
are diminished when a higher $k$ is required. As shown in 
Proposition~\ref{prop:upper_bound}, the actual cluster size will be
the maximum between the desired $k$ and the minimum size
required to fulfill $t$-closeness. Thus, because of the optimal
updating of $k$ by Algorithm~\ref{alg:t-closeness-first},
this algorithm is the one for which SSE increases the most as a 
result of the larger $k$.
Algorithms~\ref{alg:micro_merge} and~\ref{alg:k-anonymity-first},
on the other hand, are more immune to large values of $k$.
Indeed, since they prioritize the $k$-anonymous microaggregation,
the larger clusters obtained for large values of $k$
have a greater chance to already fulfill $t$-closeness
without the posterior merging step;
since $k$-anonymous clusters are created in order to 
minimize the SSE, the smaller number of merging steps required 
to fulfill $t$-closeness helps to maintain cluster homogeneity 
and avoid increasing SSE. In any case, for any value of $k$, the 
SSE for these 
algorithms is still higher than for Algorithm~\ref{alg:t-closeness-first}.

For Algorithms~\ref{alg:micro_merge} and~\ref{alg:k-anonymity-first},
it is also interesting to observe the
spikes that occur for certain values of $k$,
which are more noticeable for Algorithm~\ref{alg:micro_merge}.
Spikes occur when $k$ is not a divisor of the data set size $n$ (i.e., 1,080);
that is, when it is not possible to group all records in clusters
of size $k$.
In such cases, the microaggregation algorithm is forced to distribute
the remaining $r=n \bmod k$ records among already created clusters,
which deteriorates cluster homogeneity and thus increases SSE.
On the contrary, Algorithm~\ref{alg:t-closeness-first} is more immune to this situation,
because clusters are created to satisfy $t$-closeness, rather than
to minimize the SSE.

\section{Conclusions and research directions}
\label{sec:conclusions}
We have proposed and evaluated the use of microaggregation as a 
method to attain $k$-anonymous $t$-closeness. 

The {\em a priori} benefits of microaggregation vs generalization/recoding and 
local suppression 
have been discussed. Global recoding may recode more than needed,
 whereas local recoding complicates data analysis by mixing
together values corresponding to different levels of generalization.
Also, recoding produces a greater loss of 
granularity of the data, is more
affected by outliers, and changes numerical values to ranges.
Regarding local suppression, it complicates data analysis with 
missing values and is not obvious to combine with recoding
in order to decrease the amount of generalization. 
Microaggregation is free from all the above downsides.  

We have proposed and evaluated three different microaggregation
based algorithms to generate $k$-anonymous $t$-close data sets. 
The first one is a simple merging step that can be run after 
any microaggregation algorithm. The other
two algorithms, $k$-anonymity-first and $t$-closeness-first, take
the $t$-closeness requirement into account at the moment of   
cluster formation during microaggregation. 
The $t$-closeness-first algorithm considers $t$-closeness
earliest and provides
the best results: smallest average cluster size, smallest SSE for a given level 
of $t$-closeness, and shortest run time 
(because the actual microaggregation level is computed
beforehand according to the values of $k$ and $t$). 
Thus, {\em considering the
$t$-closeness requirement from the very beginning turns out to be 
the best option}. 

Since connections have been demonstrated between 
$t$-closeness and $\varepsilon$-differential privacy
of data sets~\cite{Soria2013differential,DomingoSoria15}, 
exploring how microaggregation could be leveraged
to implement the latter model in the case of data releases
is a natural continuation of this work. 
Moreover, we will also study the adaptation of the algorithms
to support categorical data by: i) defining an EMD suitable to
compare categorical values of different nature ({\em e.g.}, ordinal values 
such as colors, which can be sorted within a range, 
or nominal values such as jobs, hobbies, diagnoses, etc., 
which require interpreting their underlying semantics), 
ii) defining aggregation operators to compute cluster centroids 
({\em i.e.}, the categorical value that minimizes the distance to 
other values in the same cluster), and iii) properly managing records 
with numerical and categorical attributes in an integrated manner. 

\section*{Acknowledgments and disclaimer}

This work was partly supported by the European Commission 
(through projects FP7 "DwB", FP7 "Inter-Trust" and H2020 "CLARUS"), 
by the Spanish Government (through projects "ICWT" TIN2012-32757, 
"CO-PRIVACY" TIN2011-27076-C03-01 and "BallotNext" IPT-2012-0603-430000) 
and by the Government of Catalonia (under grant 2014 SGR 537). 
Josep Domingo-Ferrer is partially supported as an ICREA-Acad\`emia 
researcher by the Government of Catalonia and 
by a Google Faculty Research Award. 
Partial support by the  
Templeton World Charity Foundation is 
also acknowledged ("CO-UTILITY" grant).
The opinions expressed in this paper 
are the authors' own and do not necessarily reflect the views of the 
Templeton World Charity Foundation or UNESCO.

\ifCLASSOPTIONcaptionsoff
  \newpage
\fi



%

%


\begin{IEEEbiography}[{\includegraphics[width=1in,height=1.25in,clip]{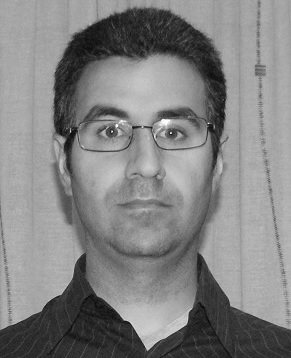}}]{Jordi Soria-Comas}
is a postdoctoral researcher at Universitat Rovira i Virgili. He has received his M. Sc. in Computer Security (2011) and Ph. D. in Computer Science (2013) degrees from the Universitat Rovira i Virgili. He also holds a M. Sc. in Finance from the Autonomous University of Barcelona (2004) and a B.Sc. in Mathematics from the University of Barcelona (2003). His research interests are in data privacy and security.
\end{IEEEbiography}

\begin{IEEEbiography}[{\includegraphics[width=1in,height=1.25in,clip,keepaspectratio]{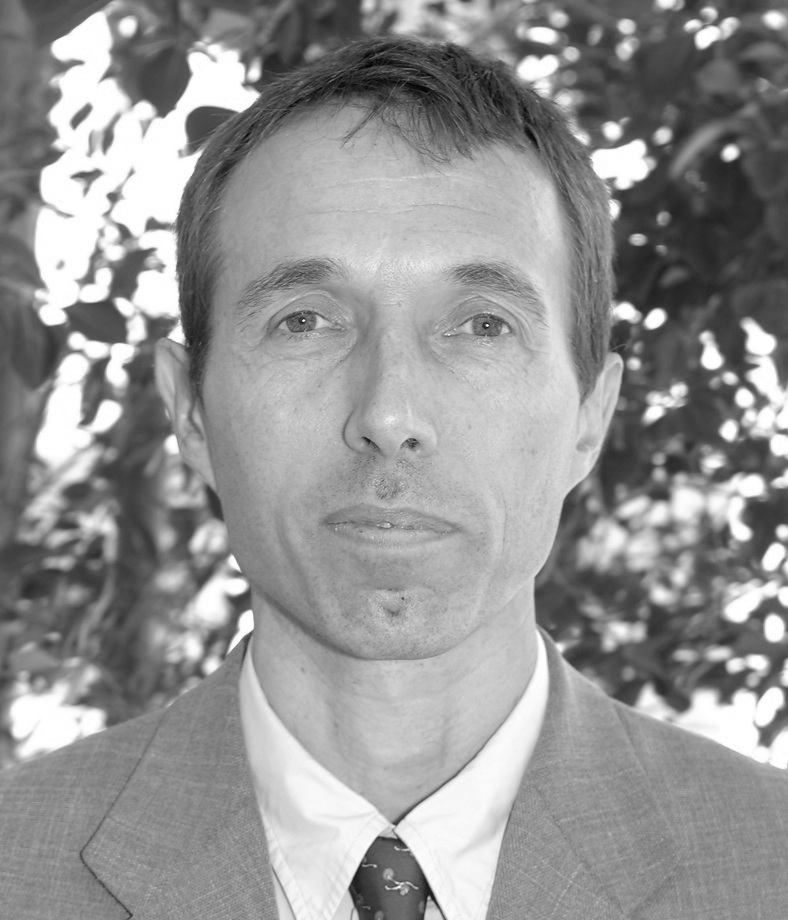}}]{Josep Domingo-Ferrer}
(Fellow, IEEE) 
is a Distinguished Professor of Computer Science and
an ICREA-Acad\`emia Researcher at Universitat Rovira i Virgili,
Tarragona, Catalonia, where he holds the UNESCO Chair in Data Privacy.
He received the MSc and
PhD degrees in Computer Science from
the Autonomous University of Barcelona in 1988 and
1991, respectively. He also holds an MSc degree in
Mathematics. 
His research interests are in data privacy, data security and cryptographic
protocols. More information on him can be found 
at \url{http://crises-deim.urv.cat/jdomingo}
\end{IEEEbiography}

\begin{IEEEbiography}[{\includegraphics[width=1in,height=1.25in,clip,keepaspectratio]{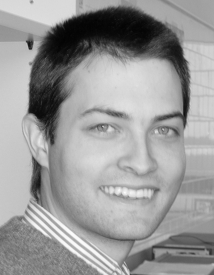}}]{David S\'anchez}
is an Associate Professor of Computer Science at Universitat Rovira i Virgili, Tarragona,
Catalonia. His research interests are in data semantics and data privacy. He received a PhD in Computer
Science from the Technical University of Catalonia. Contact him at david.sanchez@urv.cat.
\end{IEEEbiography}

\begin{IEEEbiography}[{\includegraphics[width=1in,height=1.25in,clip,keepaspectratio]{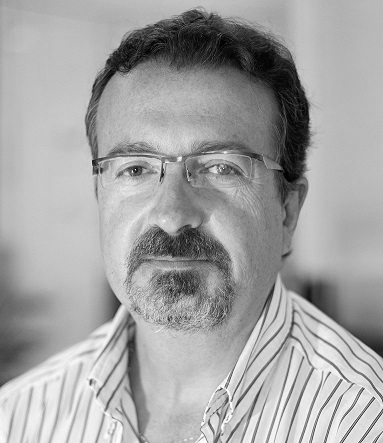}}]{Sergio Mart\'{\i}nez}
is a post-doctoral researcher at University Rovira i Virgili (URV) in Tarragona. He received an MSc in Intelligent Systems (2010) and a Ph.D in Computer 
Science (2013), both awarded by the URV. His research interests are 
in artificial intelligence, semantic similarity and privacy preservation. 
He has participated in European and Spanish research projects.
\end{IEEEbiography}




\vfill


\end{document}